\documentclass[a4paper,UKenglish,cleveref, autoref, thm-restate,authorcolumns]{lipics-v2019}

\graphicspath{ {./Figures/} {./.tmp/} }
\usepackage{cite}
\usepackage[version=4]{mhchem}
\usepackage{algorithm}
\usepackage{algpseudocode}
\theoremstyle{plain}
\newtheorem{assumption}[theorem]{Assumption}

\title{Implementing Non-Equilibrium Networks with Active Circuits of Duplex Catalysts} 

\titlerunning{Active Circuits of Duplex Catalysts} 

\author{Antti Lankinen}{Department of Bioengineering, Imperial College London, London, SW7 2AZ, United Kingdom}{antti.lankinen15@imperial.ac.uk}{}{}
\author{Ismael Mullor Ruiz}{Department of Bioengineering and Imperial College Centre for Synthetic Biology, Imperial College London, London, SW7 2AZ, United Kingdom}{i.mullor-ruiz16@imperial.ac.uk}{}{}
\author{Thomas E. Ouldridge}{Department of Bioengineering and Imperial College Centre for Synthetic Biology, Imperial College London, London, SW7 2AZ, United Kingdom}{t.ouldridge@imperial.ac.uk}{}{}

\authorrunning{A. Lankinen, I. Mullor Ruiz, and T. E. Ouldridge} 

\Copyright{Antti Lankinen, Ismael Mullor Ruiz, and Thomas Ouldridge} 



\category{} 

\relatedversion{} 

\supplement{}

\acknowledgements{}

\nolinenumbers 

\hideLIPIcs  

\begin{document}

\maketitle

\begin{abstract}
DNA strand displacement (DSD) reactions have been used to construct chemical reaction networks in which species act catalytically at the level of the overall stoichiometry of reactions. These effective catalytic reactions are typically realised through one or more of the following: many-stranded gate complexes to coordinate the catalysis, indirect interaction between the catalyst and its substrate, and the recovery of a distinct ``catalyst'' strand from the one that triggered the reaction. These facts make emulation of the out-of-equilibrium catalytic circuitry of living cells more difficult. Here, we propose a new framework for constructing catalytic DSD networks: Active Circuits of Duplex Catalysts (ACDC). ACDC components are all double-stranded complexes, with reactions occurring through 4-way strand exchange. Catalysts directly bind to their substrates, and and the ``identity'' strand of the catalyst recovered at the end of a reaction is the same molecule as the one that initiated it. We analyse the capability of the framework to implement catalytic circuits analogous to phosphorylation networks in living cells. We also propose two methods of systematically introducing mismatches within DNA strands to avoid leak reactions and introduce driving through net base pair formation. We then combine these results into a compiler to automate the process of designing DNA strands that realise any catalytic network allowed by our framework.
\end{abstract}
\newpage
\section{Introduction}
\label{sec:introduction}
DNA is an attractive engineering material due to the high specificity of Watson-Crick base pairing and well-characterised thermodynamics of DNA hybridisation \cite{santalucia_thermodynamics_2004, dirks_thermodynamic_2007}, which give DNA the most predictable and programmable interactions of any natural or synthetic molecule \cite{seeman_dna_2017}. DNA computing involves exploiting these properties to assemble computational devices made of DNA. The computational circuits are typically realised using DNA strand displacement (DSD) reactions, in which sections of DNA strands called \textit{domains} with partial or full complementarity hybridise, displacing one or more previously hybridised strands in the process \cite{zhang_dynamic_2011}. DSD is initiated by the binding of short complementary sequences called \textit{toeholds}. It is helpful to divide DSD reactions into a few common reaction steps, including: binding, unbinding, and three- or four-way strand displacement and branch migration, shown in Figure \ref{fig:dsdreactions}. DSD is an attractive scheme for computation as it can be used as a medium in which to realise chemical reaction networks (CRNs)\cite{soloveichik_dna_2010}, which provide an abstraction of systems exhibiting mass-action chemical kinetics and have been shown to be Turing complete \cite{magnasco_chemical_1997}. DSD is then Turing complete as well \cite{qian_efficient_2011, yahiro_implementation_2016}. DSD has been used to construct, for example, logic circuits \cite{seelig_enzyme-free_2006, qian_scaling_2011}, artificial neural networks \cite{qian_neural_2011, genot_scaling_2013, cherry_scaling_2018}, dynamical systems \cite{srinivas_enzyme-free_2017}, catalytic networks \cite{zhang_engineering_2007, qian_simple_2011, chen_programmable_2013}, and other computational devices \cite{adleman_molecular_1994, yin_programming_2008}. To facilitate testing and realisation of DSD systems, domain-level design tools  \cite{lakin_visual_2011, spaccasassi_logic_2019} as well as domain-to-sequence translation \cite{zadeh_nupack_2011} software have been introduced.

While DNA nanotechnology is concerned with using DNA as a non-biological material, a key goal of DNA nanotechnology is the imitation and augmentation of cellular systems. It is therefore worth considering how these natural systems typically perform computation and information processing. One ubiquitous biological paradigm for signal propagation and processing is the catalytic activation network, as exemplified by kinases \cite{marshall_map_1994,herskowitz_map_1995,manning_protein_2002}. Kinases are catalysts that modify substrates by phosphorylation and consume ATP in the process. These substrates can be, for example, transcription factors, but can also be kinases themselves that are either activated or deactivated by phosphorylation. The opposite function, dephosphorylation, is performed by phosphatases \cite{barford_structure_1998}. The emergent catalytic network then performs information propogation or computation by converting species, kinases and phosphatases, between their active and passive states. Kinase cascades are featured in many key biological functions, such as cellular growth, adhesion, and differentiation \cite{widmann_mitogen-activated_1999, manning_protein_2002} and long-term potentiation \cite{sweatt_neuronal_2001}. 

The fuel-consuming, catalytic nature of these circuits is vital in allowing them to perform functions such as signal splitting, amplification, time integration and insulation \cite{govern_energy_2014,mehta_landauer_2016,ouldridge_thermodynamics_2017,deshpande_high_2017,barton_energy_2013}. Moreover, since the key molecular species are recovered rather than consumed by reactions, catalytic networks can operate continuously, responding to stimuli as they change over time - unlike many architectures for DSD-based computation and information processing that operate by allowing the key components to be consumed \cite{qian_neural_2011,cherry_scaling_2018,adleman_molecular_1994}. This ability to operate continuously is invaluable in autonomous environments such as living cells. 

In this work, we propose a minimal mechanism for implementing reaction networks of molecules that exist in catalytically active and inactive states, a simple abstraction of natural kinase networks. In these catalytic activation networks, we implement arbitrary activation reactions of the form $A^\prime + B +\sum_i F_i \rightarrow A^\prime + B^{\prime} + \sum_i W_i$. Here, the active catalyst $A^\prime$ drives $B$ between its inactive and active states by the conversion of fuel molecules $\{F_i\}$ into waste $\{W_i\}$. Equivalent deactivation reactions in which an active catalyst deactivates a substrate are also considered. 

The rest of this paper is organised as follows. In Section \ref{sec:mole_catalysts}, we propose and motivate the concept of a direct bimolecular catalytic reaction and consider the necessary conditions for DSD species that are able to perform such reactions. Section \ref{sec:species} introduces a novel DSD framework to implement these reactions, and its computational properties are analysed in Section \ref{sec:networks}. Based on these findings, we propose a systematic method of introducing mismatched base pairs within species in our framework to improve its function in Section \ref{sec:mismatch}. We combine our findings and propositions into a software to automate the sequence-level design of any CRN that is realisable within our framework, and detail this software in Section \ref{sec:compiler}. In Section \ref{sec:discussion}, we discuss our framework, findings, and future work. We conclude the paper in Section \ref{sec:conclusion}.

\section{Direct Action of Molecular Catalysts}
\label{sec:mole_catalysts}
In kinase cascades, functional changes in substrates are a result of direct binding of the catalyst to the substrate. Moreover, the essential products of the reaction (the activated substrate and recovered catalyst) are the same molecules that initially bound to each other - albeit with some modification of certain residues, or turnover of small molecules such as ATP or ADP to which they are bound. Motivated by these facts, we propose the following definition for a direct bimolecular catalytic activation reaction.

\begin{definition}[Direct bimolecular catalytic activation]{\label{def:catalyticreaction}}
Consider the (non-elementary) reaction
\begin{align*}{\label{eq:catalyticreactiondefinition}}
    {A^\prime + B +\sum_i F_i \rightarrow A^\prime + B^{\prime} + \sum_i W_i},
\end{align*}
where $A^\prime$ catalyses the conversion of inactive $B$ to active $B^\prime $, using ancillary fuels $\{F_i\}$ and producing waste $\{W_i\}$. The overall reaction is a direct bimolecular catalytic activation reaction if and only if:
\begin{enumerate}
    \item The reaction is initialised with the interaction of $A^\prime$ and $B$.
   \item The $A^\prime$ and $B$ molecules have molecular cores that are retained in the products $A^\prime$ and $B^\prime$, rather than the input molecules being consumed and distinct outputs released. 
\end{enumerate}
Deactivation reactions have an equivalent form, but convert $B^\prime$ to $B$.
If the same overall reaction stoichiometry is implemented differently, the reaction is a \normalfont{pseudocatalytic bimolecular activation reaction}. 
\end{definition}

\begin{figure}
 \centering
 \begin{subfigure}[b]{0.2 \textwidth}
  \centering
  \includegraphics[width=\textwidth]{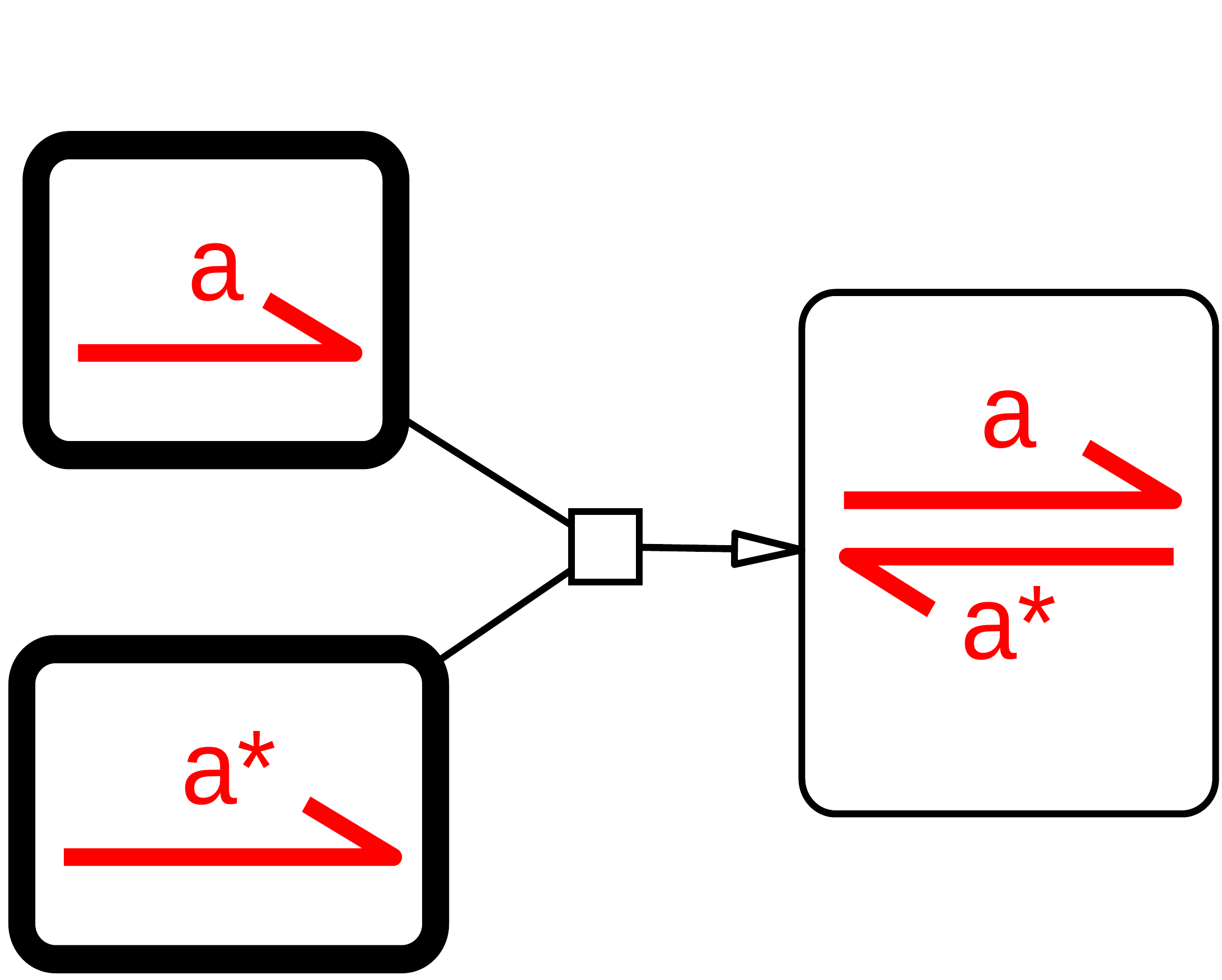}
  \caption{Bind}
 \end{subfigure}
 \begin{subfigure}[b]{0.2 \textwidth}
  \centering
  \includegraphics[width=\textwidth]{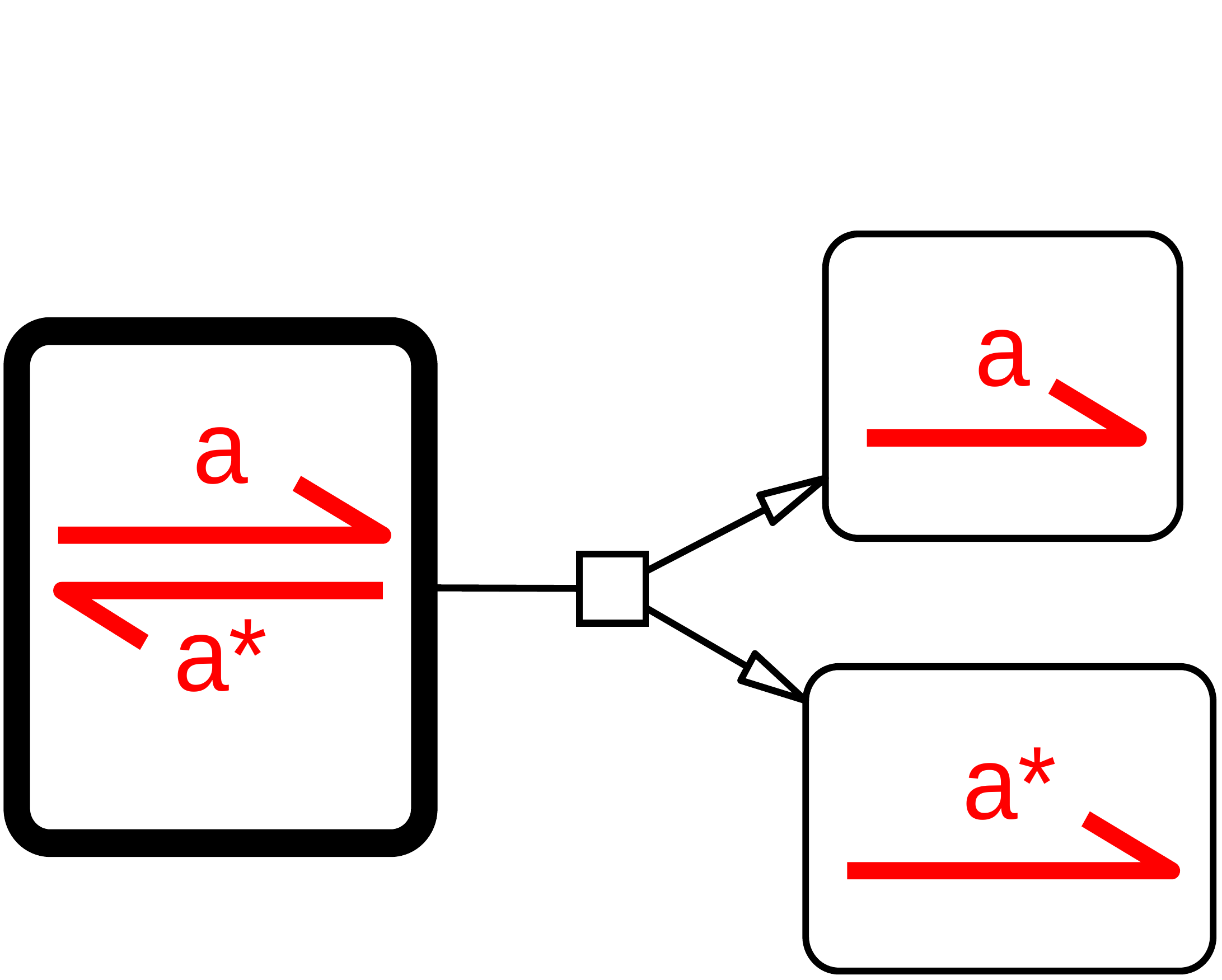}
  \caption{Unbind}
 \end{subfigure}
 \begin{subfigure}[b]{0.22 \textwidth}
  \centering
  \includegraphics[width=\textwidth]{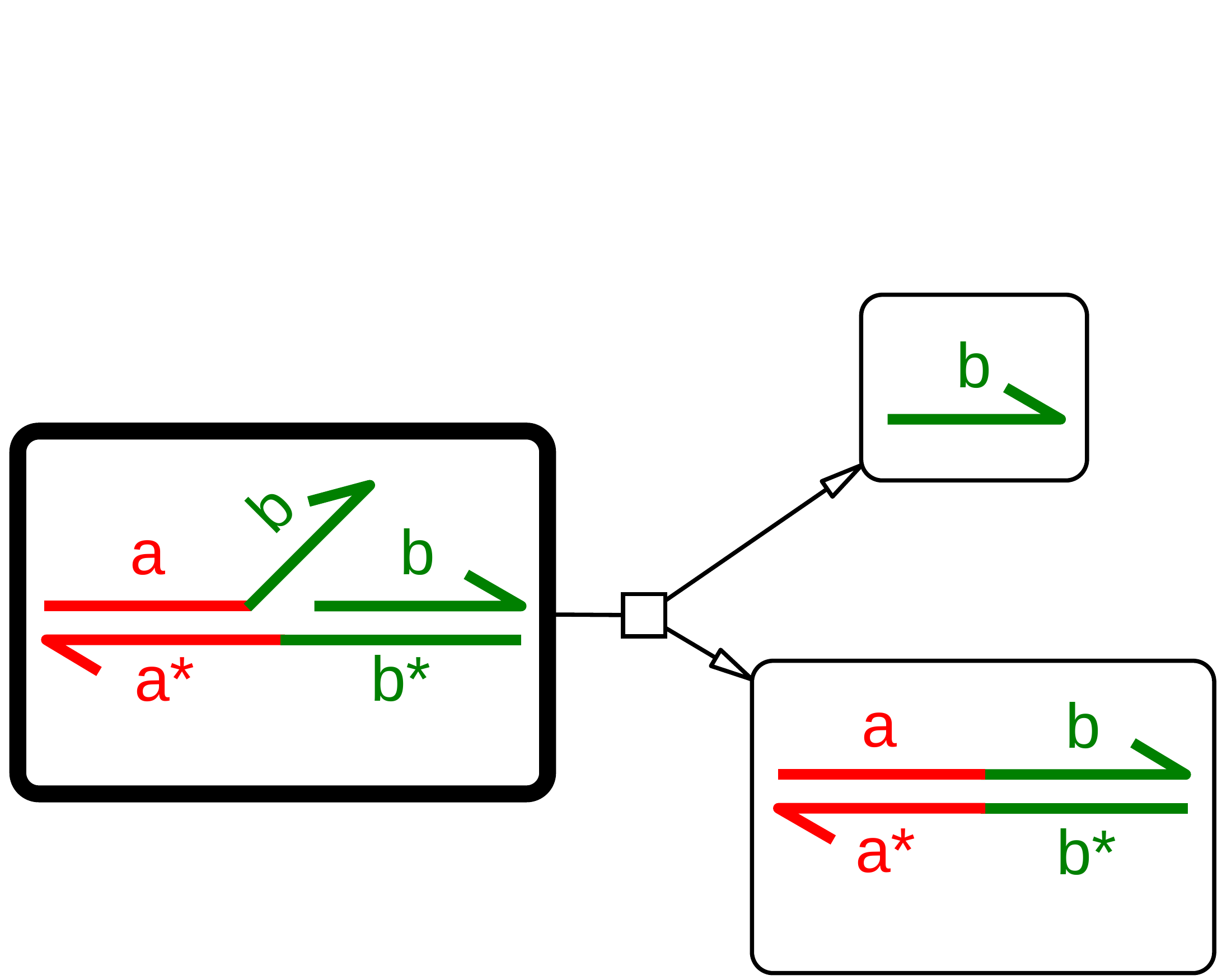}
  \caption{Displace (3-way)}
 \end{subfigure}
 \begin{subfigure}[b]{0.33 \textwidth}
  \centering
  \includegraphics[width=\textwidth]{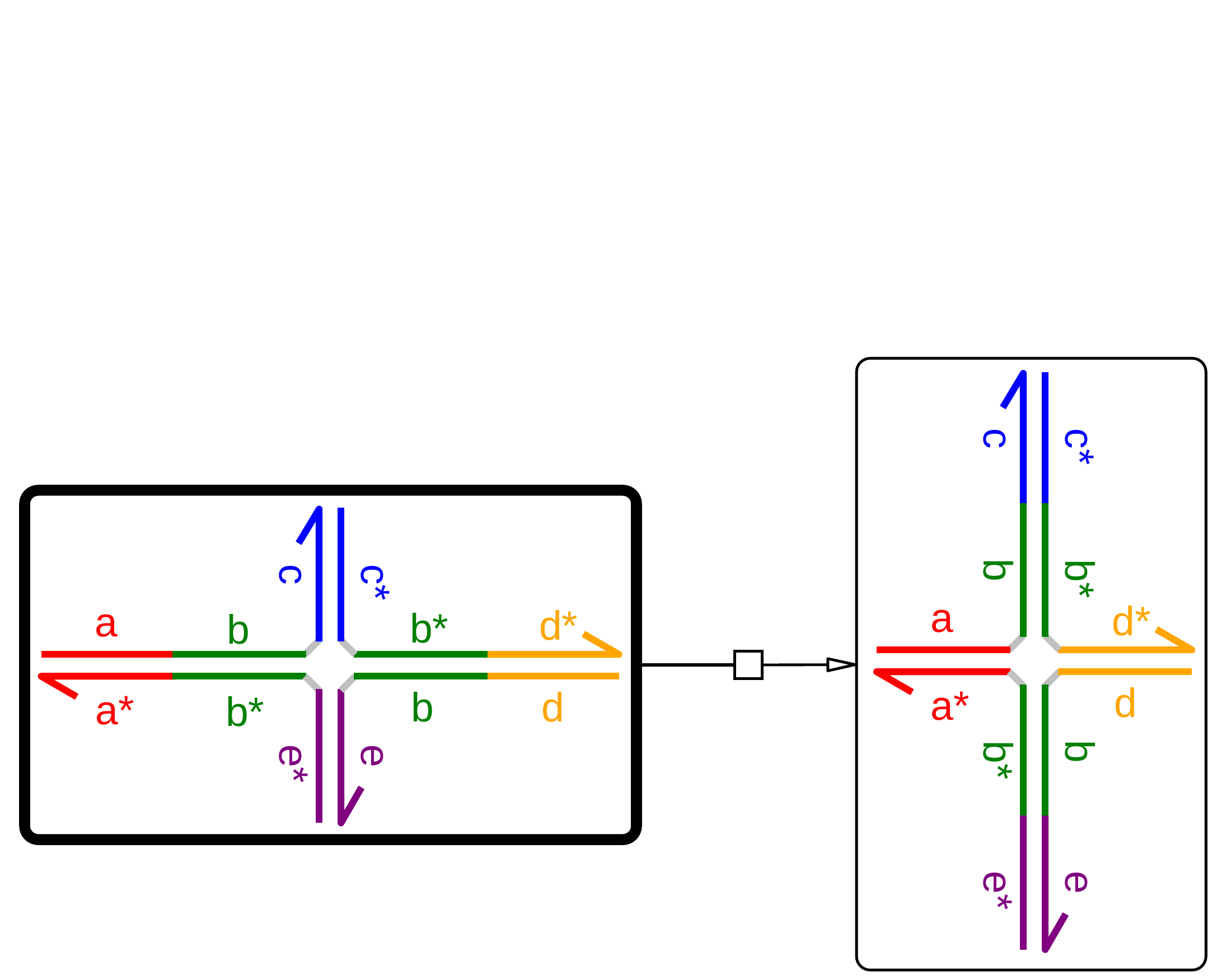}
  \caption{Branch migrate (4-way)}
 \end{subfigure}
 \caption{Basic reaction steps in the DSD formalism, as represented by Visual DSD \cite{lakin_visual_2011}. Each domain is represented by a letter and a colour. "*" denotes the Watson-Crick complement. The barbed end of a strand indicates the 3' end.}
 \label{fig:dsdreactions}
\end{figure}

Direct bimolecular catalytic (de)activation reactions have some important functional properties. The first is that, if the first step of the reaction requires the presence of $A^\prime$ and $B$, nothing can happen unless both molecules are present. In pseudocatalytic implementations, as we discuss below, it is possible to produce activated $B^\prime$ or sequester $A^\prime$ even if no $B$ is present, violating the logic of activation-based networks. The second is that the persistence of a molecular core of both the substrate and the catalyst allows either or both to be localised on a surface or scaffold, as is observed for some kinase cascades in living cells \cite{elion_ste5_1995,whitmarsh_mammalian_1998,schaeffer_mp1_1998} and is often proposed for DNA-based systems \cite{qian_parallel_2014, teichmann_robustness_2014, ruiz_connecting_2015, bui_localized_2018, chatterjee_spatially_2017}.

A number of DNA computing frameworks have been developed to implement reactions of the stoichiometry of Definition \ref{def:catalyticreaction}. The simplest, illustrated in Figure \ref{fig:seesaw}\,(a), involves a two-step seesaw gate \cite{qian_simple_2011,haley_design_2020}. An input molecule ($A^\prime$ in Definition \ref{def:catalyticreaction}) binds to a gate-output complex ($F$), releasing the output ($B^\prime$). The input is then displaced by a molecule conventionally described as the fuel, but fulfilling the role of $B$ from Definition \ref{def:catalyticreaction} in the context of catalysis, recovering $A^\prime$ and generating a waste duplex ($W$). Although the $A^\prime$ strand recovered at the end of the process is the same one that initiated the process, the $B$ and $B^\prime$ molecules are distinct and the reaction is not initiated by the binding of $A^\prime$ and $B$; it is therefore pseudocatalytic. 

This pseudocatalysis can have important consequences. If a small quantity of input $A^\prime$ is added to a solution containing the gate-output complex $F$ but no $B$, a large fraction of $A^\prime$ is sequestered and a corresponding amount of $B^\prime$ is produced. This sequestration of $A^\prime$ and production of $B^\prime$ from nothing violates the logic of ideal catalytic activation networks.

More complex strategies to implement reactions of the stoichiometry of Definition \ref{def:catalyticreaction}
using DSD exist \cite{chen_programmable_2013,qian_efficient_2011}. 
These approaches rely on the catalyst and substrate ($A^\prime$ and $B$ from Definition \ref{def:catalyticreaction}) interacting with a gate, rather than binding to each other, and the recovered catalyst and product are separate strands - the reactions are therefore pseudocatalytic. In certain limits, these strategies can approximate a mass-action dependence of reaction rates on the concentrations of $A^\prime$ and $B$ \cite{chen_programmable_2013, plesa_stochastic_2018}, providing a better approximation to the logic of ideal catalytic activation circuits than the simple seesaw motif. The price, however, is the need to construct large multi-stranded gate complexes to facilitate the reaction; the complexity of these motifs is a major barrier to implementing such systems in autonomous setting such as living cells. Moreover, localising catalysts and substrates to a scaffold or surface remains challenging when the molecules themselves are not recovered.

\begin{figure}
 \centering
 \includegraphics[width=0.6 \textwidth]{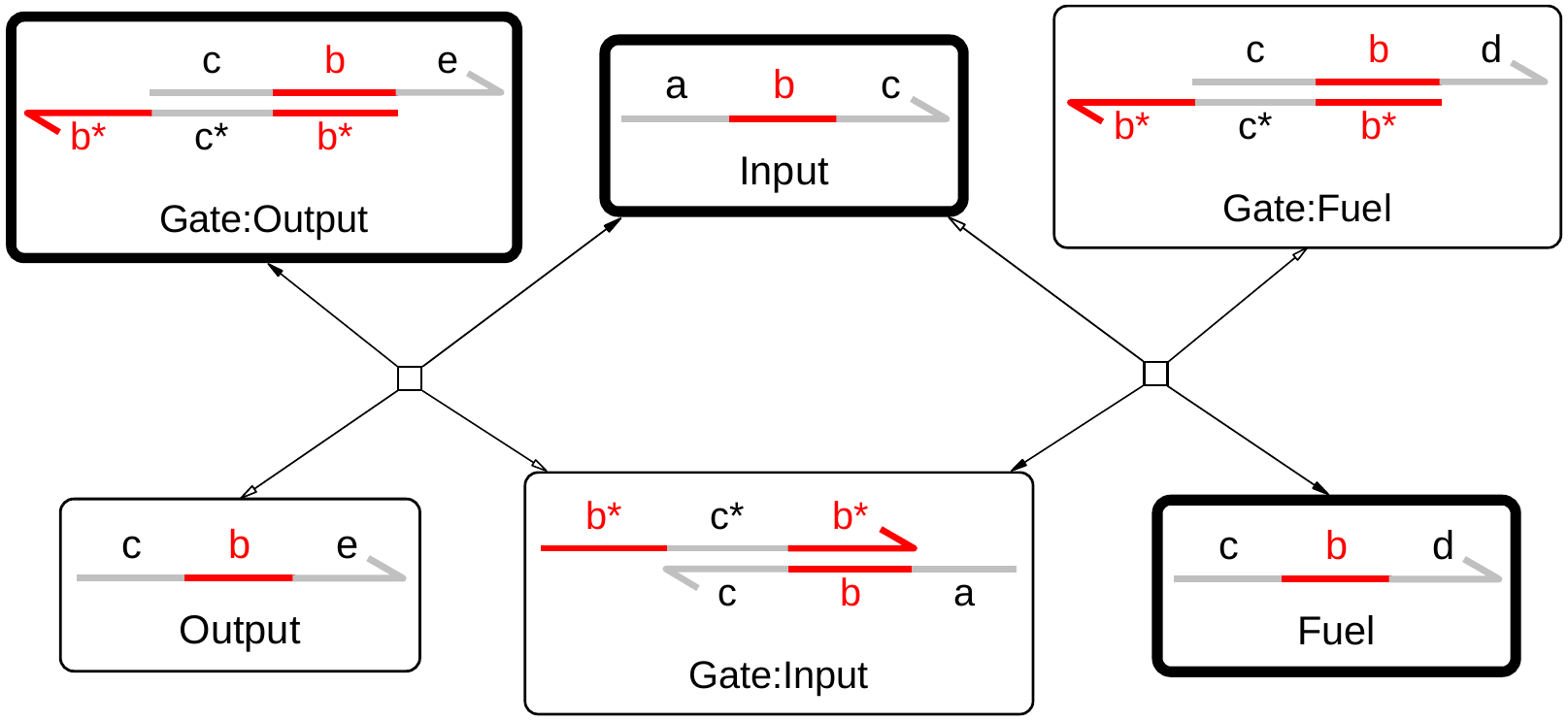}
 \caption{Catalytic reaction using a seesaw gate \cite{qian_simple_2011,haley_design_2020}. Reactants are shown in bold boxes; the input acts pseudocatalytically to ``convert" the fuel into an output, with ancillary gate complexes consumed and produced. Each compound reaction is illustrated by a small square, and consists of sequential bind, displace, and unbind reactions. All reactions are reversible; open arrows indicate reactions proceeding forwards, and closed arrows by reactions proceeding backwards.}
 \label{fig:seesaw}
\end{figure}

We now consider how to design minimal DSD-based units that implement direct biomolecular catalytic (de)activation in catalytic activation networks. If the core of the substrate species $B$ must be retained in the product $B^\prime$, $B$ and $B^\prime$ cannot simply be two strands with a slightly different sequence. Instead, $B$ and $B^\prime$ must either be distinct complexes of strands, in which at least one strand is common, or have different secondary structure within a single strand, or both. To avoid complexities in balancing the thermodynamics of hairpin loop formation with bimolecular association, and suppressing the kinetics of unimolecular rearrangement, we do not pursue the possibility of engineering metastable secondary structure within a strand. At least one of $B$ and $B^\prime$ must therefore consist of at least two strands. Moreover, since each activation state of each species must be a viable substrate in an arbitrary catalytic (de)activation network, the simplest approach that allows for a generic catalytic mechanism is to implement all substrate/catalyst species as two-stranded complexes. 

\section{ACDC: A Duplex-Based Catalytic DSD Framework}
\label{sec:species}
We introduce the Active Circuits of Duplex Catalysts (ACDC) scheme to implement catalytic activation networks through direct bimolecular catalytic (de)activation. Each reaction has three inputs: a substrate, a catalyst, and a single fuel complex. The outputs are a modified substrate, the recovered catalyst and a waste complex. The domain-level structures of these species are shown in  Figure \ref{fig:species}.

Substrates and catalysts -- hereafter referred to as \textit{major species} -- are anatomically identical. Each consists of two strands, each of which has one central long domain ($\sim$ 20 nucleotides (nt)) and two toeholds ($\sim$ 5 nt) on each side of the long domain. In major species, these strands are called the \textit{identity strand} and the \textit{state strand}. The identity strand is the preserved molecular core; the state strand specifies the activation state of a major species at a particular time (specifically, through the domain at its $5^\prime$ end - labelled ``a'' in Fig.~\ref{fig:species}). 

The two strands in a major species are bound by three central domains; the \textit{outer toeholds} at either end of the strands are \textit{available} (unbound). 
Major species thus contain two \textit{interfaces} at either end of the molecule, both displaying two available toeholds, one on each constituent strand. The \textit{inner toeholds}, which are bound in major species, are described as \textit{hidden}. We call the interface at the 5’ end of the state strand and the 3’ end of the identity strand the \textit{downstream} interface and the interface with the 3’ end of the state strand and 5’ end of the identity strand the \textit{upstream} interface.

All other two-stranded species in ACDC, including fuel and waste species, are described as \textit{ancillary species}. They have a distinct structure from major species, but are identical to each other (Figure \ref{fig:species}). Ancillary species also consist of two strands of five domains, but are bound by the central long domain and two shorter flanking toeholds (one outer toehold and one inner toehold) on one side. They therefore possess just one interface of available toeholds, but this interface presents two contiguous available toeholds on each strand. 

\begin{figure}
 \centering
 \begin{subfigure}[b]{0.5 \textwidth}
  \centering
  \includegraphics[width=\textwidth]{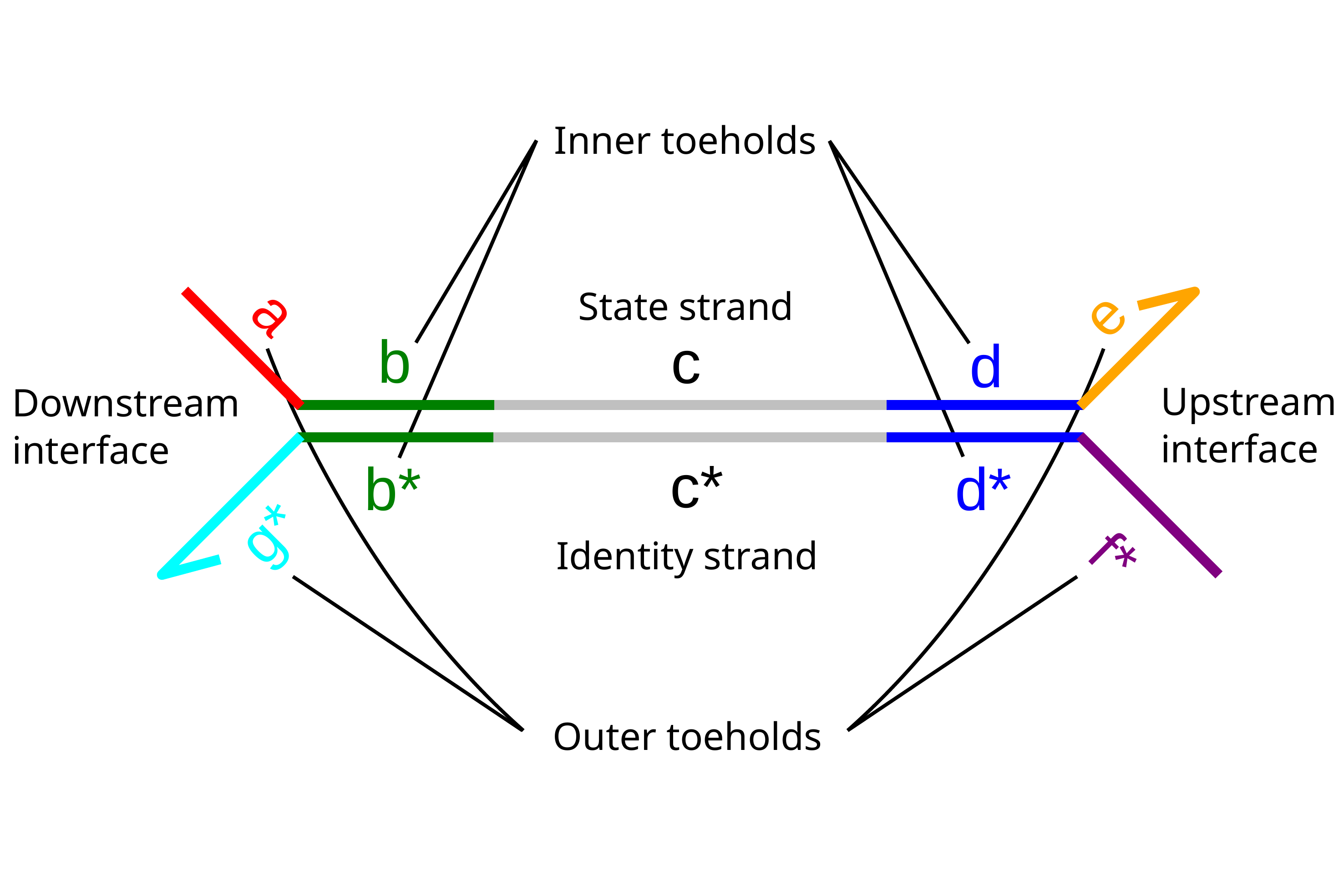}
  \caption{Major species}
 \end{subfigure}
 \begin{subfigure}[b]{0.4 \textwidth}
  \centering
  \includegraphics[width=\textwidth]{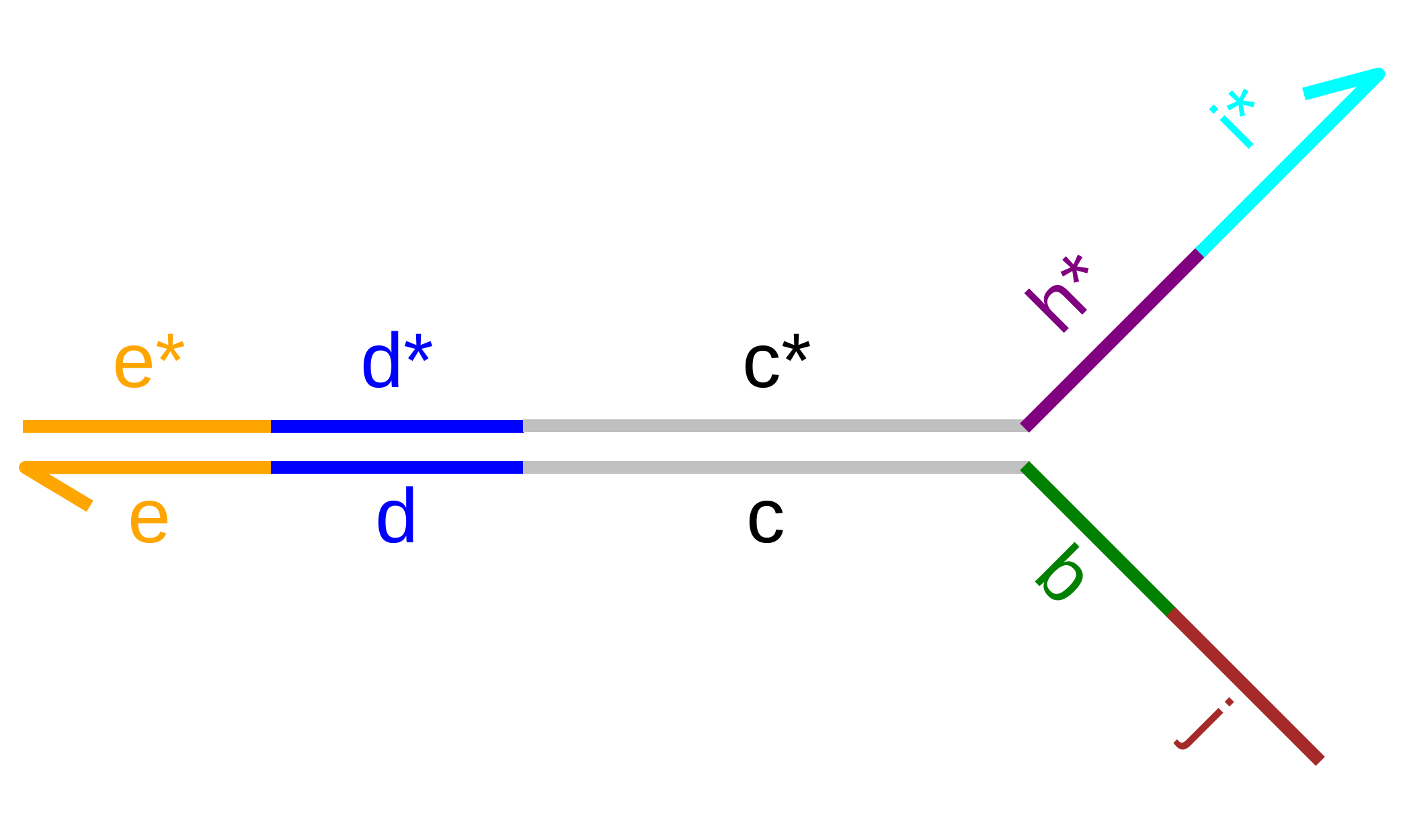}
  \caption{Ancillary species}
 \end{subfigure}
 \caption{(a) Topology of major species in the ACDC system (substrates or catalysts), illustrating upstream and downstream interfaces, and inner and outer toeholds. The long central domain forms a stable binding duplex. (b) Topology of ancillary species (fuel, waste or substrate-catalyst complex).}
 \label{fig:species}
\end{figure}

The catalytic reaction of a single ACDC  unit proceeds as shown in Figure \ref{fig:reaction}. The downstream interface of the catalyst $A^\prime$ and upstream interface of the substrate $B$ bind together through recognition of all four available toeholds in the relevant interfaces. The resultant complex undergoes a 4-way branch migration, with the base pairs between the state and identity strand of the substrate and catalyst being exchanged for base pairs between the two state strands and the two identity strands. 
After the exchange of a hidden toehold and the central binding domain, the 4-stranded complex is held together by only two inner toeholds on either side of a 4-way junction. Dissociation by spontaneous detachment of these toeholds  creates two ancillary product species, a waste $W_{AB\rightarrow B^\prime}$ and an intermediate complex $AB$. The sequence of these three reactions is called the \textit{2r-4} reaction \cite{johnson_impossibility_2019}.

The fuel $F_{AB\rightarrow B^\prime}$ is identical to the waste, except for a single toehold. This toehold corresponds to the outer toehold of the state strand of $B$ from the downstream interface. $F_{AB\rightarrow B^\prime}$ and $AB$ can undergo another \textit{2r-4} reaction, producing $B^\prime$ ($B$, but with a single domain changed in the downstream interface) and recovering the catalyst. With the downstream interface of substrate $B$ changed into that of $B^\prime$, the substrate has been activated and could act as a catalyst to another reaction, provided that an appropriate downstream substrate and fuel were present. An equivalent catalytic process could trigger another reaction converting $B'$ to $B$, \textit{deactivating} $B$, analogous to dephosphorylation by a phosphatase. 

The basic ACDC unit in Figure \ref{fig:reaction} satisfies the criteria for direct bimolecular catalytic activation, since the reaction is initiated by the binding of $A^\prime$ and $B$, and the identity strands in the major species are retained throughout. ACDC relies on the experimentally-verified mechanism of toehold-mediated 4-way branch migration \cite{venkataraman_autonomous_2007, dabby_synthetic_2013,lin_hierarchical_2018, kotani_multi-arm_2017}. The number of base pairs and complexes is unchanged by each \textit{2r-4} reaction, and therefore a bias for clockwise activation cycles (as opposed to anticlockwise deactivation) would require a large excess of fuel complexes $F_{AB\rightarrow B^\prime}$ relative to waste $W_{AB\rightarrow B^\prime}$. In addition, for a single catalytic cycle to operate as intended, the following assumptions must hold:

\begin{assumption}[Stability of complexes]
\label{ass:stability}
It is assumed that strands bound together by long domains are stable and will not spontaneously dissociate. It is also assumed that if two strands are bound by a pair of complementary domains, any adjacent pairs of complementary domains that could bind to form a contiguous duplex are not available.  
\end{assumption}

\begin{assumption}[Detachment of products]
\label{ass:detach}
It is assumed that 4-stranded complexes bound together by two pairs of toehold domains either side of a junction can dissociate into duplexes.   
\end{assumption}

\begin{assumption}[Need for two complementary toeholds to trigger branch migration]
\label{ass:specificity}
It is assumed that if a 4-stranded complex is formed by the binding of a single pair of toehold domains, it will dissociate into product duplexes, rather than undergo branch migration. 
\end{assumption}

Assumption \ref{ass:stability} ensures that the system keeps its duplex-based structure, and that toeholds are well hidden in complexes when required.
Assumption \ref{ass:detach} is necessary to avoid all species being sequestered into 4-stranded complexes. Note that the assumption is not that detachment must happen extremely quickly, since such 4-stranded complexes need to be metastable enough to initiate branch migration with reasonable frequency. It is equivalent to the need for single toeholds to detach in 3-way toehold exchange reactions \cite{qian_simple_2011}. In practice, toehold length and conditions such as temperature could be tuned to optimize the relative propensity for branch migration and detachment. Given a reasonable balance between branch migration and detachment, Assumption \ref{ass:specificity} -- which enables the switching of $B$ and $B^\prime$ to have a downstream effect -- is also likely to be satisfied.

\begin{figure}[t]
 \centering
 \includegraphics[width=\textwidth]{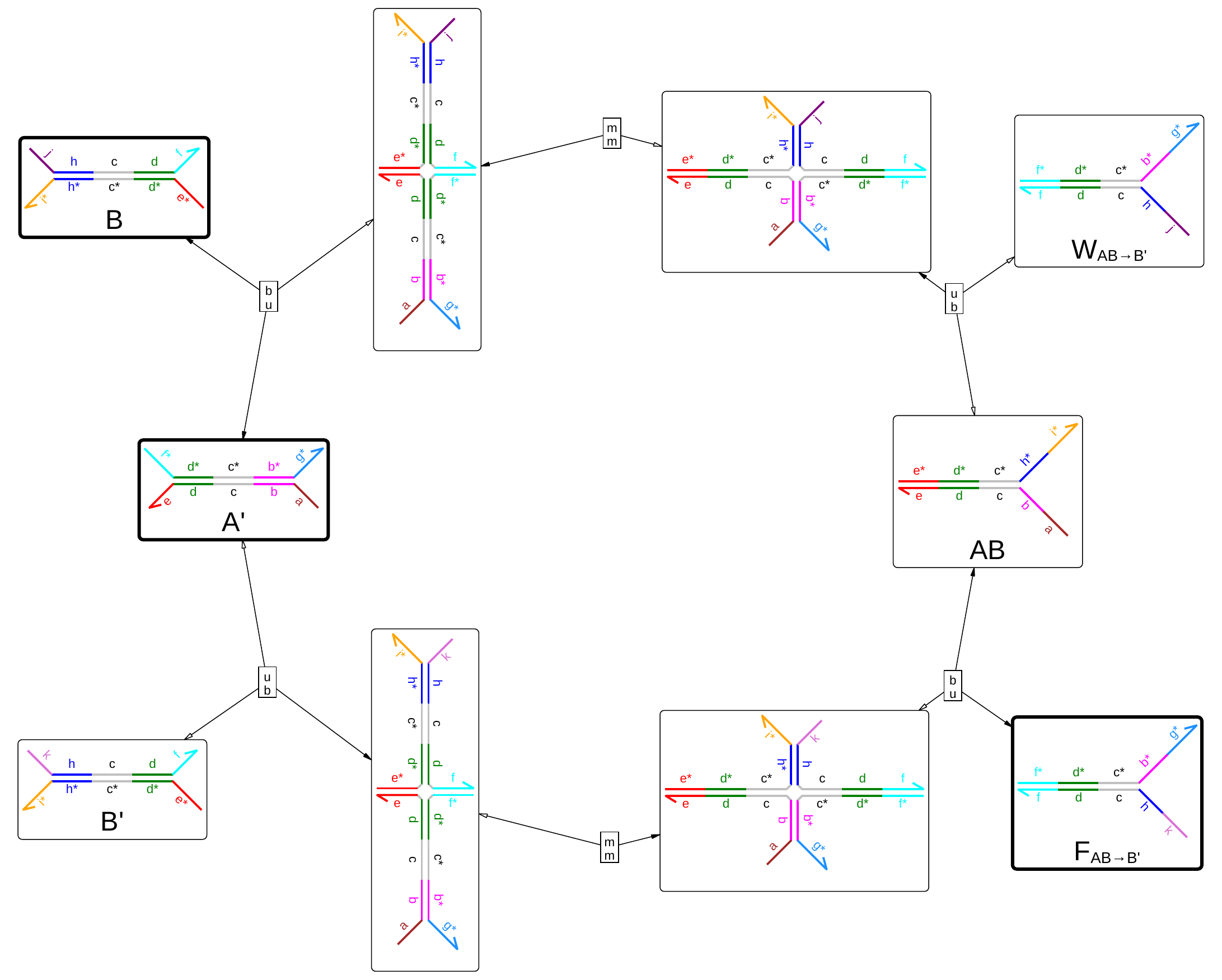}
 \caption{A basic ACDC reaction unit $A^\prime + B + F_{AB \rightarrow B^\prime} \rightarrow A^\prime + B^{\prime} +  W_{AB \rightarrow B^\prime}$, as represented by Visual DSD \cite{lakin_visual_2011}. Inputs to the reaction are shown in bold, and each small box corresponding to a reaction step is labelled with b/u (bind/unbind) or m (migrate). Imbalances in the concentration of fuel and waste drive the reaction clockwise (the direction indicated by open arrows).}
 \label{fig:reaction}
\end{figure}

\section{Domain-based constraints in ACDC Networks}
\label{sec:networks}
Larger catalytic activation networks can be constructed from the basic ACDC units of Figure \ref{fig:reaction}, since the activated substrate $B^\prime$ can itself act as a catalyst. Let $A \rightarrow B$ be a shorthand for the reaction ${{A'} + B + F_{AB\rightarrow B^\prime}\rightarrow {A'} + {B'} + W_{AB\rightarrow B^\prime}}$ and $C \dashv B$ a shorthand for the reaction ${C'} + {B'} + F_{CB^\prime \rightarrow B} \rightarrow {C} + B +W_{CB^\prime \rightarrow B}$. Then, any potential catalytic activation network can be represented as a weighted directed graph, where nodes represent catalyst/substrate species and edges represent activation (edge weight 1) or deactivation (edge weight -1). Is it possible to realise any such graph using ACDC? 

\begin{figure}
 \centering
 \begin{subfigure}[t]{0.18 \textwidth}
  \centering
  \includegraphics[width=\textwidth]{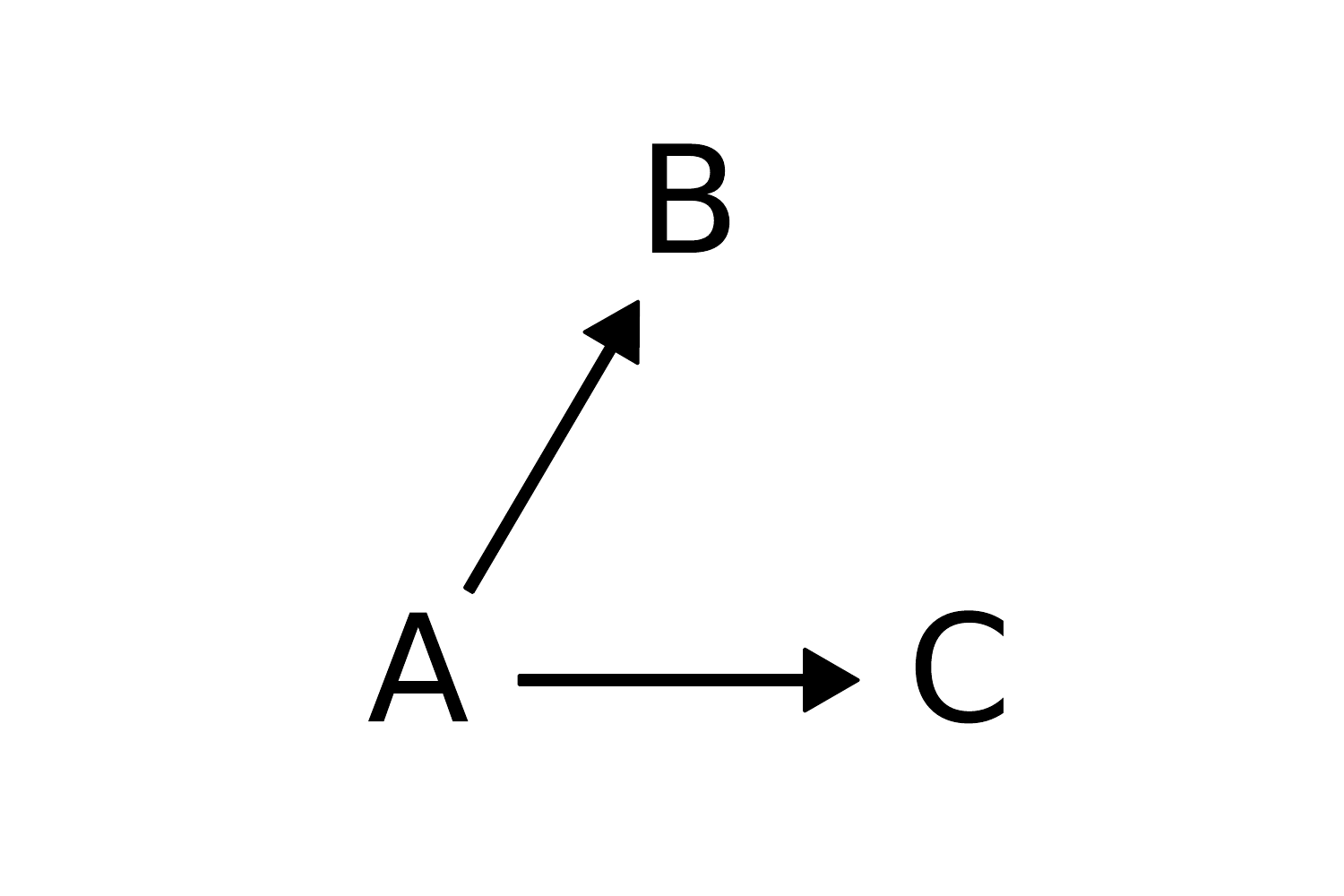}
  \caption{Split}
 \end{subfigure}
 \begin{subfigure}[t]{0.18 \textwidth}
  \centering
  \includegraphics[width=\textwidth]{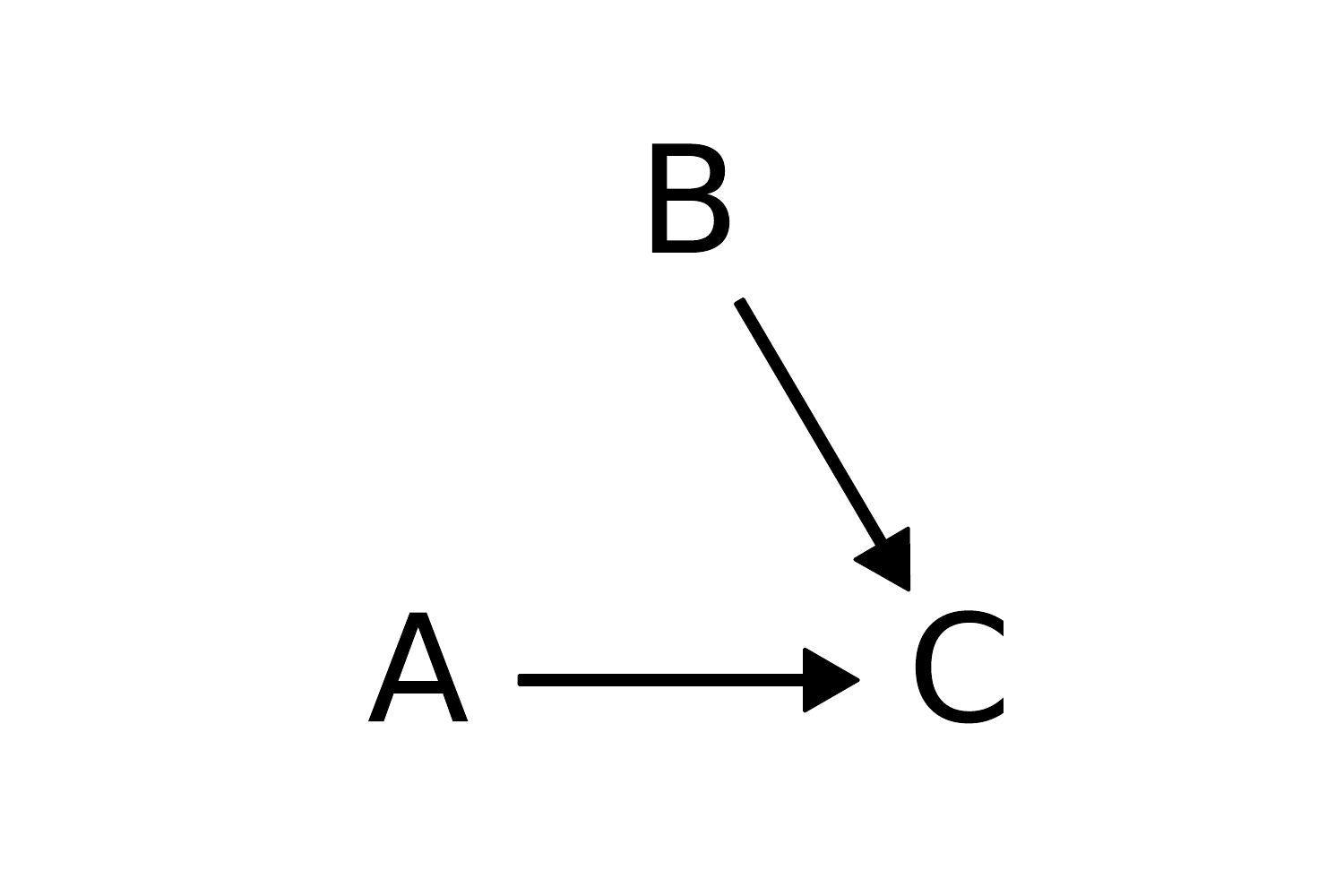}
  \caption{Integrate}
 \end{subfigure}
 \begin{subfigure}[t]{0.18 \textwidth}
  \centering
  \includegraphics[width=\textwidth]{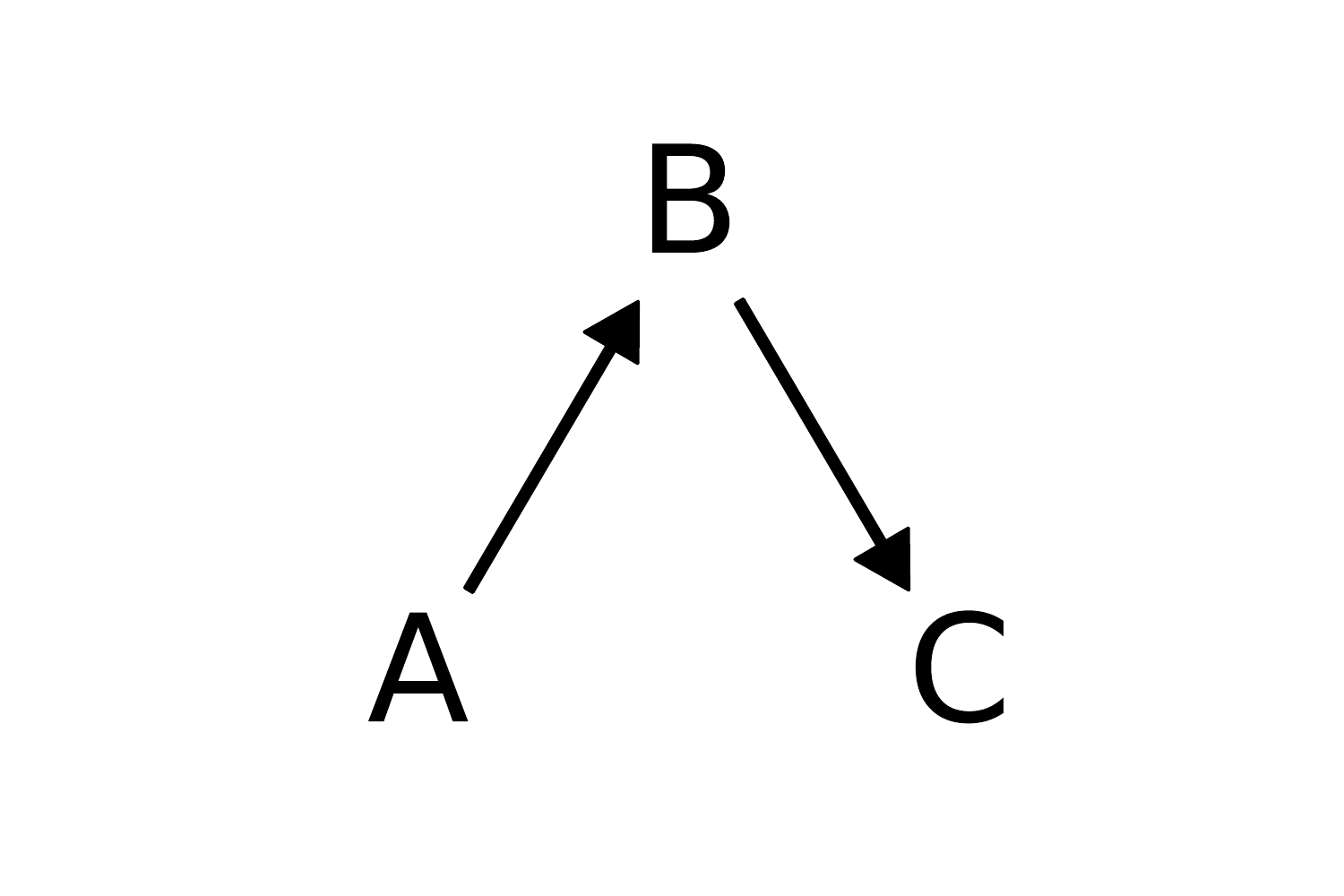}
  \caption{Cascade}
 \end{subfigure}

 \begin{subfigure}[t]{0.18 \textwidth}
  \centering
  \includegraphics[width=\textwidth]{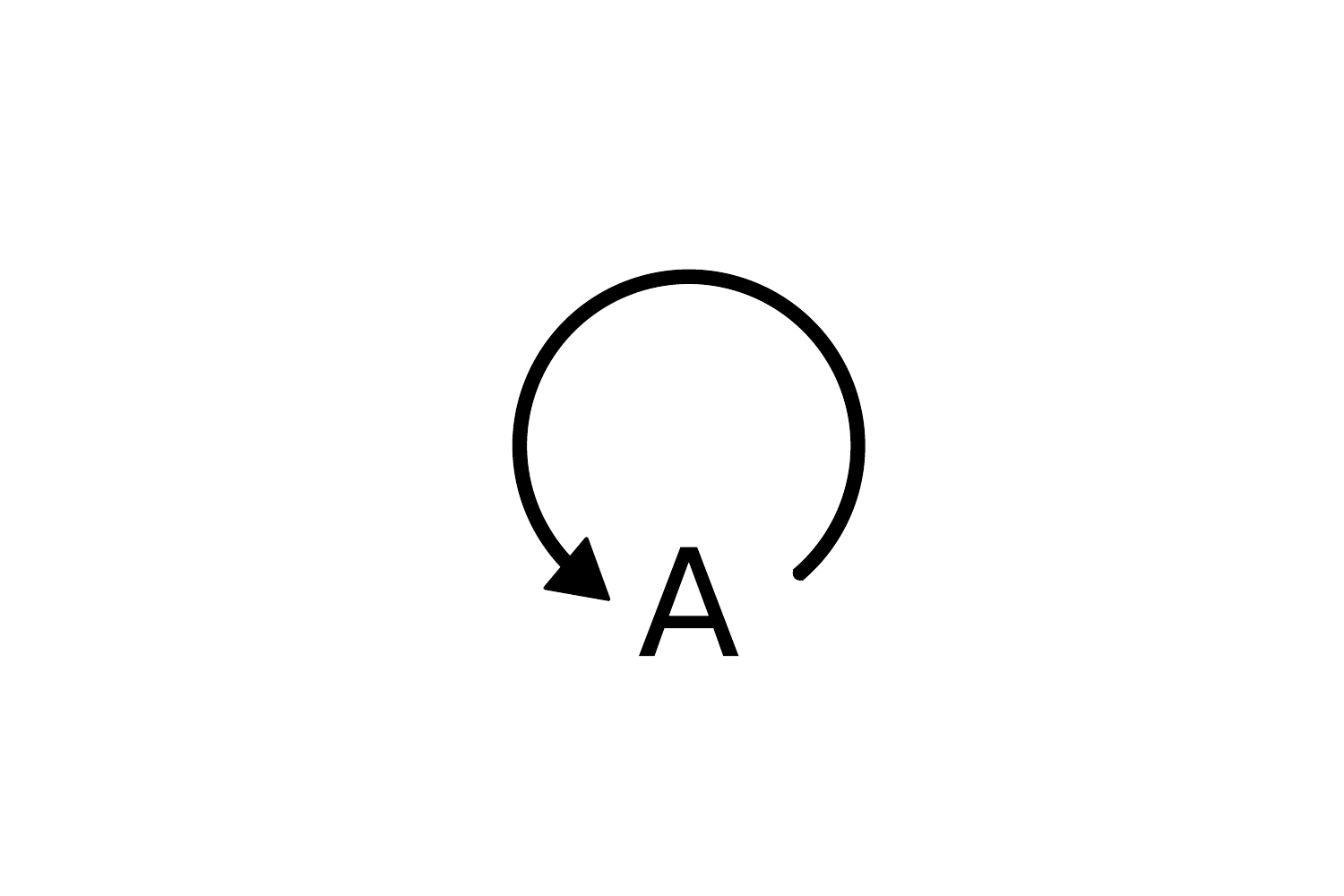}
  \caption{Auto-activation loop}
 \end{subfigure}
 \begin{subfigure}[t]{0.18 \textwidth}
  \centering
  \includegraphics[width=\textwidth]{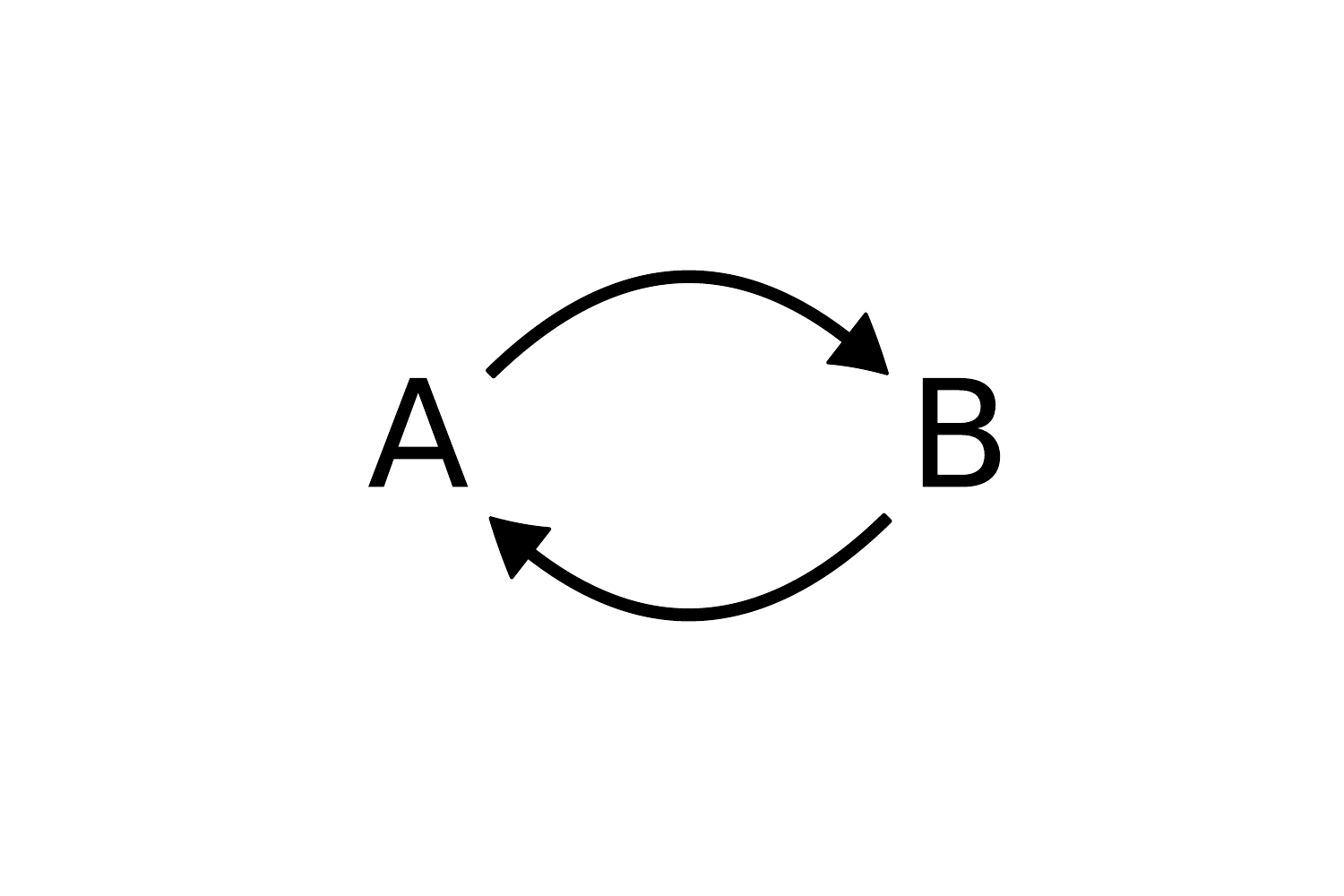}
  \caption{Bidirectional edge}
 \end{subfigure}
 \begin{subfigure}[t]{0.18 \textwidth}
  \centering
  \includegraphics[width=\textwidth]{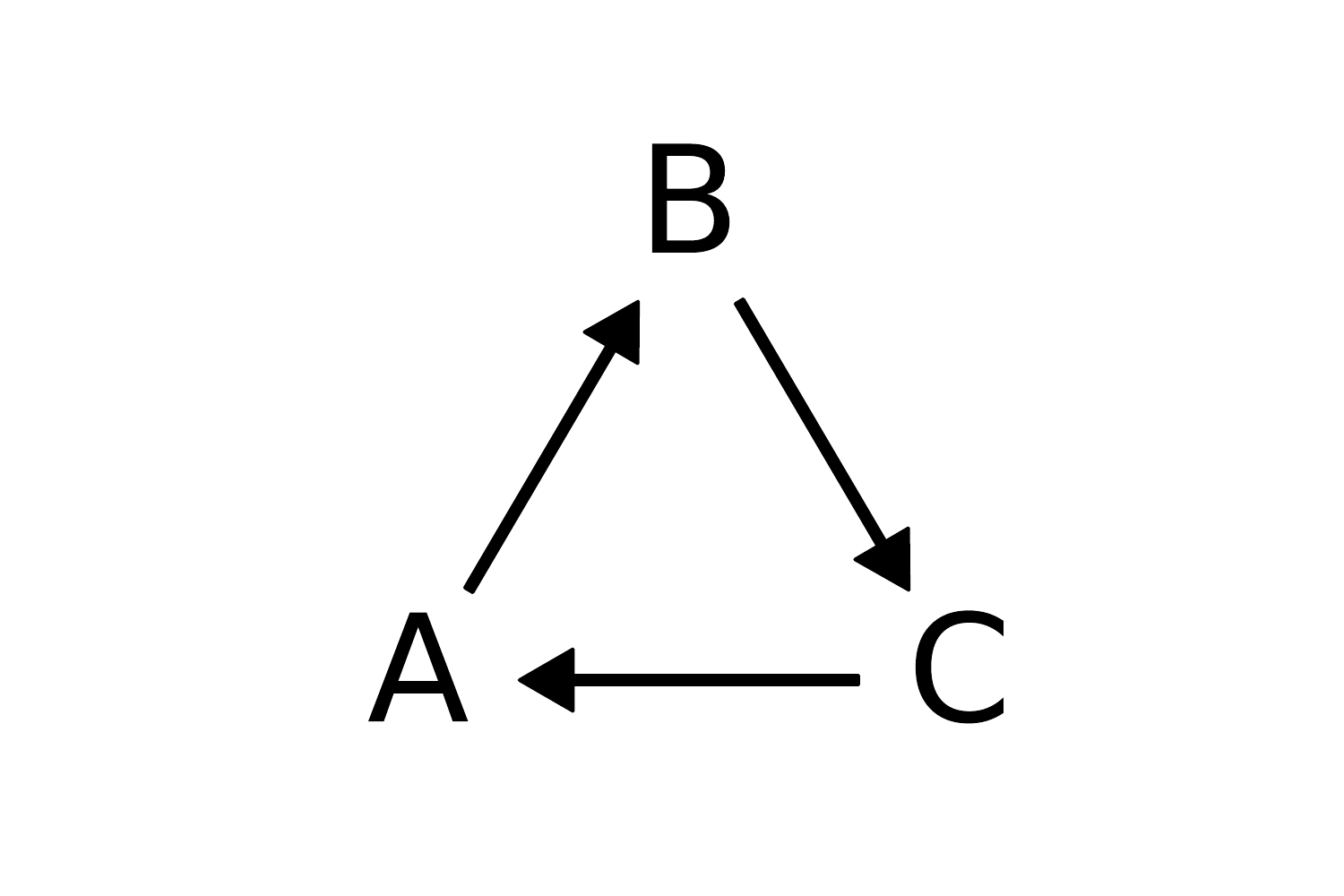}
  \caption{Feedback loop}
 \end{subfigure}
 \begin{subfigure}[t]{0.18 \textwidth}
  \centering
  \includegraphics[width=\textwidth]{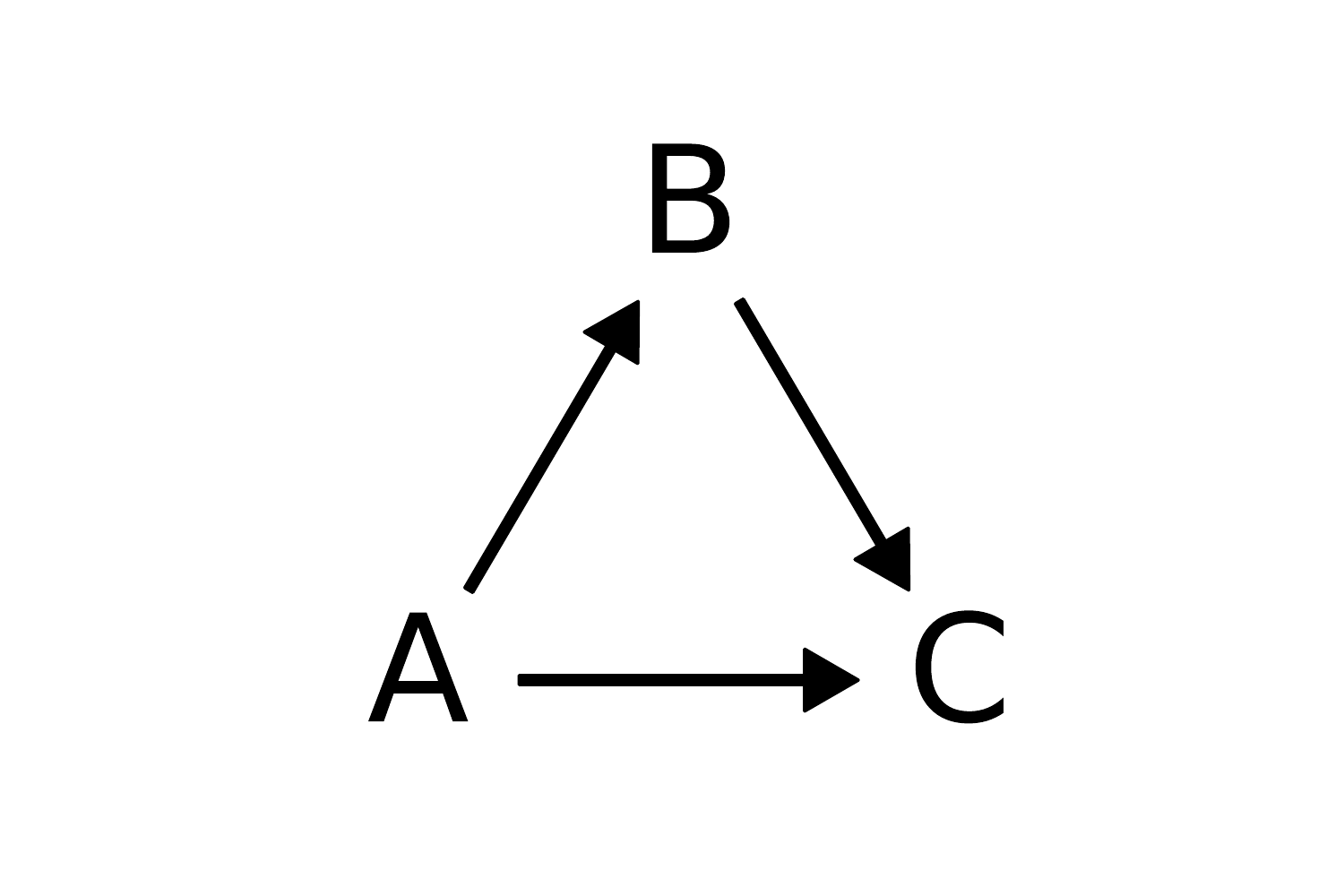}
  \caption{Feedforward loop}
 \end{subfigure}
 \caption{Minimal example motifs of interest in a catalytic activation network.}
 \label{fig:motifs}
\end{figure}


\begin{assumption}[Toehold orthogonality]
{\label{ass:orthogonality}}
We assume that there are sufficiently many toehold domain sequences that cross-talk between non-complementary domains is negligible. 
\end{assumption}

Since ACDC components share a long central domain, specificity is entirely driven through toehold recognition. As noted by Johnson, \cite{johnson_impossibility_2019}, there is a finite number of orthogonal short toehold domains that limits the size of the connected network that can be constructed. We assume that the network of interest does not violate this limit. We instead ask the realisability question at the level of domains.

\begin{definition}[Realisability]
\label{def:realisable}
 A catalytic activation network is realisable using the ACDC framework if a domain structure for a set of strands can be specified such that: 
 \begin{enumerate} 
    \item All network reactions are represented by a basic ACDC unit. 
    \item No two species possess domains that allow a \textit{2r-4} reaction (a full four-way strand exchange) that preserves the number of bound domains and is initiated by the binding of two available and complementary pairs of toeholds, unless the reaction is part of an ACDC unit representing a reaction in the network.
    \item No two strands can form an uninterrupted duplex of four bound domains or more.
    \item No two species (including all wastes, fuels and catalyst-substrate complexes) possess two available toehold pairs that could form a contiguous complementary duplex.
 \end{enumerate}
\end{definition}

Condition 2 rules out reactions that respect the architecture of ACDC, but which involve reactants that are not intended to interact. Condition 3 rules out strand exchange reactions that  allow an increase in the number of bound domains, which would sequester additional toeholds and violate the ACDC architecture (it is assumed that strand exchange reactions that would reduce the number of bound domains can be neglected). Condition 4 rules out the formation of 4-stranded complexes that can only dissociate by disrupting an uninterrupted two-toehold duplex. Contiguous duplexes of this kind are potentially stable, even if they cannot undergo strand exchange, and would potentially sequester components.

\begin{theorem}[Realisability with  activation implies realisability with deactivation ]{\label{thm:equivalence}}
If a catalytic activation network with purely activation reactions is realisable using the basic ACDC formalism, it is also realisable using the basic ACDC formalism if any subset of those reactions are converted to deactivation.  
\end{theorem}
\begin{proof}
 A deactivation reaction is simply an activation reaction with the role of the fuel and waste reversed. Therefore a domain structure specification that realises a given network with activation reactions also realises all networks of the same structure.
\end{proof}

\subsection{Realisability of Motifs in the ACDC formalism}

\begin{figure}
 \centering
 \begin{subfigure}[t]{0.85 \textwidth}
  \centering
  \includegraphics[width=\textwidth]{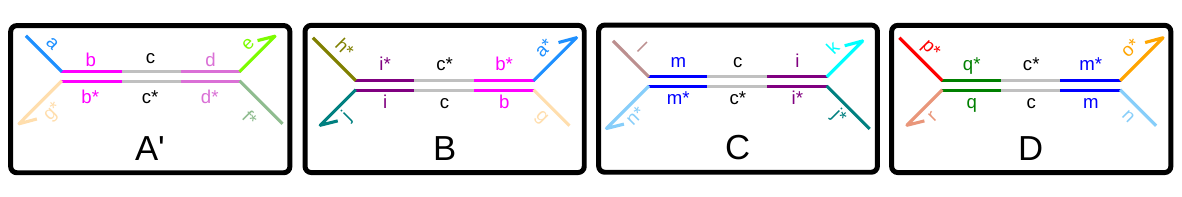}
  \caption{Major species}
 \end{subfigure}
 \begin{subfigure}[t]{0.9 \textwidth}
  \centering
  \includegraphics[width=\textwidth]{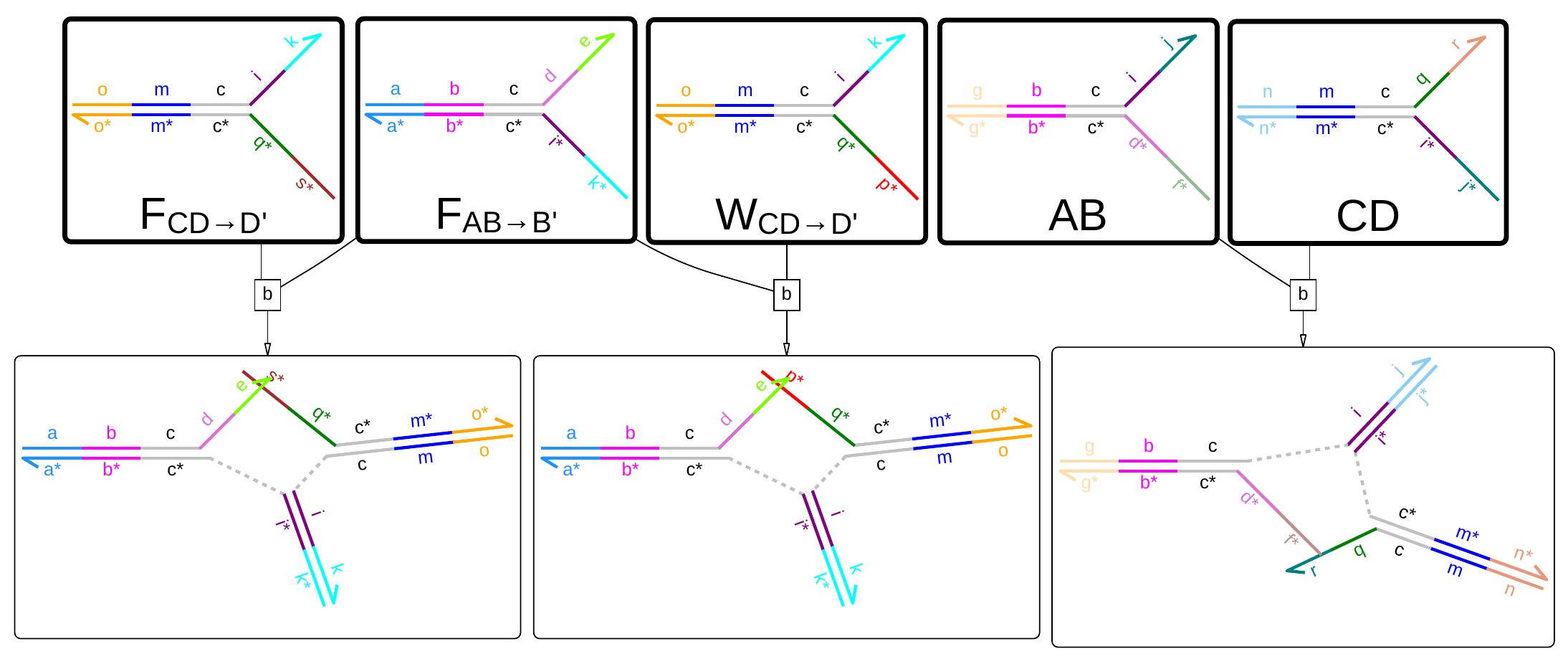}
  \caption{Ancillary species and unwanted reactions}
 \end{subfigure}
 \caption{Major species and a subset of ancillary species from an implementation of $A \rightarrow B \rightarrow C \rightarrow D$ using the ACDC formalism. Three unwanted reactions occur between the shown ancillary species.}
 \label{fig:leaks}
\end{figure}

Since there are infinitely many networks, we restrict our analysis to a set of motifs (generalised versions of the minimal examples depicted in Figure  \ref{fig:motifs}), establishing whether these motifs can be realised in isolation. 
The \textit{split}, \textit{integrate} 
 \textit{cascade}, \textit{self-activation}, \textit{bidirectional edge}, \textit{feedback loop} (FBL), and \textit{feedforward loop} (FFL) are chosen because of their importance in biology and synthetic biology \cite{alon_introduction_2019,gardner_construction_2000,elowitz_synthetic_2000}. The proofs of theorems not explicitly given in this section are provided in Appendix \ref{app:proofs}.

\subsubsection{Motifs Without Loops}

 Theorems \ref{thm:split} and \ref{thm:integrate} establish that arbitrarily complex split and integrate motifs are realisable. \par
\begin{theorem}[Split motifs are realisable]{\label{thm:split}}
 Consider the set of $N$ reactions
 \begin{align*}
   A \rightarrow B_1 \hspace{5mm}
   A \rightarrow B_2 \hspace{5mm}
 \hdots \hspace{5mm}
 &A \rightarrow B_N,
 \end{align*}
 in which all $B_i$ are distinct nodes from $A$.
 Such a network is realisable for any $N \geq 1$.
\end{theorem}
\begin{theorem}[Integrate motifs are realisable]{\label{thm:integrate}}
 Consider the set of $N$ reactions
 \begin{align*}
   A_1 \rightarrow B \hspace{5mm}
   A_2 \rightarrow B \hspace{5mm}
 \hdots \hspace{5mm}
 &A_2 \rightarrow B,
 \end{align*}
 in which all $A_i$ are distinct nodes from $B$. Such a system is realisable
 for any $N \geq 1$.
\end{theorem}

Although all networks consist of simply combining split and integrate motifs for each node, proving that all split and integrate motifs are realisable in isolation does not prove that any network assembled from them is realisable.  We therefore explore other simple motifs. For example, consider the \textit{cascade} motif (a 3-component example is illustrated in Figure \ref{fig:motifs}). 

\begin{lemma}[The ancillary species of a catalyst's upstream reactions and substrate's  downstream reactions cause leak reactions]
Consider a reaction $B \rightarrow C$, and further assume that $A \rightarrow B$ and $C \rightarrow D$ for at least one species $A$ and at least one species $D$. Then  $AB$ and $CD$, and $F_{AB\rightarrow B^\prime}$ and $F_{CD\rightarrow D^\prime}$/$W_{CD\rightarrow D^\prime}$ possess two available toehold pairs that could form a contiguous complementary duplex. No other violations of realisability occur.
\label{lemma:leak1}
\end{lemma}
An example is shown in Fig.\ref{fig:leaks}. The essence of the problem is that both the inner and outer toehold domains from the downstream end of $B^\prime$ are available in $AB$ and $F_{AB\rightarrow B^\prime}$, and the inner and outer toehold domains from the upstream end of $C$ are available in $CD$, $F_{CD\rightarrow D^\prime}$ and $W_{CD\rightarrow D^\prime}$. Since the downstream end of $B$ is complementary to the upstream end of $C$, the result is that the species can bind to each other strongly.

\begin{theorem}[Cascades with $N\geq 4$ components are not realisable]
Consider the set of $N$ reactions
$A_1 \rightarrow A_{2},\, A_{2} \rightarrow A_{3} \,...\, A_{N-1}\rightarrow A_N$, in which all $A_i$ are distinct. For $N\geq 4$, this network is not realisable.  
\label{thm:cascade}
\end{theorem}
\begin{proof}
 A direct consequence of Lemma \ref{lemma:leak1} and Definition \ref{def:realisable}.
\end{proof}

\begin{theorem}[Cascades with $N\leq 3$ components are realisable]{\label{thm:cascade2}}
The set of reactions
$A_1 \rightarrow A_{2},\, A_{1} \rightarrow A_{2},\, A_{2}\rightarrow A_3$, in which all $A_i$ are distinct, is realisable.    
\end{theorem}
\begin{proof}
 A direct consequence of Lemma \ref{lemma:leak1} and Definition \ref{def:realisable}.
\end{proof}

\begin{theorem}[Long cascades are non-realisable due to a particular type of leak reaction only]
Consider the set of $N$ reactions
$A_1 \rightarrow A_{2},\, A_{2} \rightarrow A_{3} \,...\, A_{N-1}\rightarrow A_N$, in which all $A_i$ are distinct. This network would be realisable if reactions between ancillary species $A_iA_{i+1}$ and $A_{i+2}A_{i+3}$, and  $F_{A_{i}A_{i+1}\rightarrow A_{i+1}^\prime}$ and $F_{A_{i+2}A_{i+3}\rightarrow A_{i+3}^\prime}$/$W_{A_{i+2}A_{i+3}\rightarrow A_{i+3}^\prime}$,  were absent. 
\label{thm:cascade3}
\end{theorem}

The result of Theorem \ref{thm:cascade} is discouraging, since cascades are a major feature of kinase networks \cite{marshall_map_1994, herskowitz_map_1995}. Nonetheless, we will continue the analysis of remaining motifs, and present a potential solution in Section \ref{sec:mismatch}.

\subsubsection{Motifs With Loops}
A network possesses a loop if it is possible to traverse a path that begins and ends at the same node without using the same edge twice. For the purposes of this classification, a given (directed) edge can be traversed in either direction. Loops are common components of natural networks, providing the possibility of oscillation, bistability and filtering \cite{de_ronde_multiplexing_2014, alon_introduction_2019}.

\begin{theorem}[Loops of odd length are not realisable]
\label{thm:odd_loops}
Consider a system of reactions $A_1 \leftrightarrow A_2 \leftrightarrow A_3 \hdots A_{N-1} \leftrightarrow A_1$, where $\leftrightarrow$ indicates a catlytic activation in either direction. This network is not realizable if $N$ is odd, unless the long central domain is self-complementary.
\end{theorem}
\begin{proof}
 ACDC circuits require that the long central domain alternates between a sequence and its complement in the identity strands of catalysts and their substrates. If $N$ is odd, then the sequence must be self-complementary for this alternation to happen.
\end{proof}
Introducing a self-complementary central domain is a strategy that risks a competition between duplexes and single-stranded hairpins. We do not consider it further.  

\begin{theorem}[Self interactions and bidirectional edges are not realisable]
{\label{thm:self_bi}}
Consider a system of reactions $A_1 \rightarrow A_2 \rightarrow A_3 \hdots A_{N-1} \rightarrow A_1$. This network is not realisable if $N\leq2$.
\end{theorem}
The ACDC system is not inherently suited to auto-activation or bidirectional interactions. These motifs require complementarity between both the downstream and upstream toeholds of either a single species, or two species. Strands in the system therefore violate condition 3 of Definition \ref{def:realisable} and will tend to hybridise to form fully complementary duplexes.

An isolated feedback loop is a network of size $N$ with a single directed path around the network. A simple example of length 3 is shown in Fig. \ref{fig:motifs}(f).
\begin{theorem}[Feedback loops are not realisable]
{\label{thm:fbl}}
  Consider the feedback loop
$A_1 \rightarrow A_2 \rightarrow A_3 \hdots A_{N-1} \rightarrow A_1$.
 Such a system is not realisable for any $N$.
\end{theorem}
\begin{proof}
 A direct consequence of Theorems \ref{thm:cascade}, \ref{thm:odd_loops}, and \ref{thm:self_bi}.
\end{proof}
As a consequence of Theorems \ref{thm:odd_loops} and \ref{thm:self_bi}, any realisable feedback loop must have $N\geq4$. However, a feedback loop of this length faces the same issues as a cascade: formation of stable, undesired products between ancillary species. As with cascades, the problem is essentially local, due to interactions between ancillary species in reaction $n$ and reaction $n+2$.
\begin{theorem}[Long feedback loops with an even number of units are non-realisable due to a particular type of leak reaction only]
{\label{thm:fbl}}
  Consider the feedback loop
 \begin{align*}
  &A_1 \rightarrow A_2 \hspace{5mm} A_2 \rightarrow A_3 \hspace{5mm} \hdots \hspace{5mm} A_{N-1} \rightarrow B_N \hspace{5mm} A_N \rightarrow A_1 
  \end{align*}
For $N$ even, $N\geq 4$, this network would be realisable if reactions between ancillary species $A_iA_{i+1}$ and $A_{i+2}A_{i+3}$, and  $F_{A_{i}A_{i+1}\rightarrow A_{i+1}^\prime}$ and $F_{A_{i+2}A_{i+3}\rightarrow A_{i+3}^\prime}$/$W_{A_{i+2}A_{i+3}\rightarrow A_{i+3}^\prime}$,  were absent.  Here, the index $j$ in $A_j$ should be interpreted modularly: $A_j = A_{j-N}$ for $j>N$.       
\end{theorem}

An isolated feedforward loop is a network of size $N$ with two directed paths from one node $i$ to another node $j$. Every other node appears exactly once in one of these paths. An example with path lengths of 1 and 2 is shown in Figure \ref{fig:motifs}.

\begin{theorem}[The relative lengths of paths are constrained in feedforward loops]{\label{thm:ffl}}
 Consider the generalised feedforward loop
 \begin{align*}
  &A \rightarrow B_1 \hspace{5mm} B_1 \rightarrow B_2 \hspace{5mm} \hdots \hspace{5mm} B_{N-1} \rightarrow B_N \hspace{5mm} B_N \rightarrow D \\
  &A \rightarrow C_1 \hspace{5mm} C_1 \rightarrow C_2 \hspace{5mm} \hdots \hspace{5mm} C_{M-1} \rightarrow C_M \hspace{5mm} C_M \rightarrow D 
 \end{align*}
  For such a network to be realisable, it is necessary that $N \geq 1$, $M \geq 1$, and $N-M$ is even. 
\end{theorem}
The constraint on the relative length of loops arises from Theorem \ref{thm:odd_loops}. Feedforward loops involving paths with no intermediates are not realisable due to the existence of unintended strand exchange reactions within the path that contains intermediates. 

Since each path in a feedforward loop is a cascade, Theorems \ref{thm:cascade} and \ref{thm:ffl} imply that only feedforward loops with a single intermediate in each branch are realisable.

\begin{theorem}[Realisability of feedforward loops]
\label{thm:ffl2}
Consider the generalised feedforward loop
 \begin{align*}
  &A \rightarrow B_1 \hspace{5mm} B_1 \rightarrow B_2 \hspace{5mm} \hdots \hspace{5mm} B_{N-1} \rightarrow B_N \hspace{5mm} B_N \rightarrow D \\
  &A \rightarrow C_1 \hspace{5mm} C_1 \rightarrow C_2 \hspace{5mm} \hdots \hspace{5mm} C_{M-1} \rightarrow C_M \hspace{5mm} C_M \rightarrow D 
 \end{align*}
 Such a system is realisable if and only if $N = 1$ and $M = 1$.
\end{theorem}
\begin{proof}
 As a consequence of Theorems \ref{thm:split}, \ref{thm:integrate}, \ref{thm:cascade}, \ref{thm:cascade2}, and \ref{thm:ffl}, all other FFLs are not realisable. The realisability of the FFL with N=1 and M=1 can be verified by inspection.
\end{proof}

Typically, feedforward loops use branches of different lengths to achieve a complex response to a signal over time \cite{de_ronde_multiplexing_2014, alon_introduction_2019}. Such networks are not realisable. Indeed, our analysis of various motifs has revealed that the 
majority are not realisable. Broadly speaking, there are a number of small motifs (eg. auto-activation, bi-directional reactions, feedforward loops with no intermediates in one branch) that cannot be achieved because the major species themselves interact directly. In addition, loops of odd total length are not realisable due to the nature of complementary base pairs. However, most motifs are ruled out because of a single type of interaction, between the ancillary species in one reaction and the ancillary species in another reaction that occurs two steps downstream. In Section \ref{sec:mismatch}, we propose a strategy to overcome this last problem, massively increasing the scope of the ACDC framework.

\section{Overcoming the Cascade Leak Reaction and Introducing Hidden Thermodynamic Drive}
\label{sec:mismatch}
The most severe limitation of the ACDC system detailed in Section \ref{sec:species} is expressed by Theorem \ref{thm:cascade}. Long cascades, and loops incorporating cascades, are non-realisable due to interactions between ancillary species of a given reaction, and ancillary species of a reaction separated by two catalytic steps (Theorem \ref{thm:cascade3}).

\begin{assumption}[Mismatches destabilise complexes held together by two contiguous toehold domains]
\label{ass:mismatches}
We assume that a single mismatched C-C or G-G base pair, positioned adjacent to the interface of two toehold domains, is sufficiently destabilizing that an unwanted complex formed only by the binding of these toehold domains no longer precludes realisability. 
\end{assumption}
The basic design of the ACDC motif assumes that toehold binding is relatively weak; two toehold domains on either side of a junction must be able to dissociate by Assumption \ref{ass:stability}. Individual C-C or G-G mismatches are known to be highly destabilising \cite{santalucia_thermodynamics_2004}, and should similarly allow for two contiguous domains to detach. Given Assumption \ref{ass:mismatches}, the challenge is then to systematically introduce mismatches so that all interactions between ancillary species identified in Theorem \ref{thm:cascade3} are compromised by a mismatch, without compromising intended circuit activity. Our full scheme is visualised in Figure \ref{fig:mismatches}.

\begin{definition}[Mismatches proposed to destabilize unintended complexes]
We propose the following mismatches.
\begin{enumerate}
    \item We propose that the upstream interface of every major species is made distinct for active and inactive states. Specifically, we introduce a G base at the inner edge of the outer toehold domain of the state strand of the inactive species, and a C base in the same position for the active species. Catalysts that (de)activate that species possess a C(G) in the complementary position of their downstream interface. 
    \item We introduce a C-C mismatch at the outer edge of the inner toehold domain at the downstream interface of each major species. This mismatch is eliminated in the formation of waste complexes, and retained in the substrate-catalyst complexes. 
\end{enumerate}
\label{def:miss1}
\end{definition}

\begin{assumption}[Mismatches cannot cause leak reactions]{\label{ass:mismatch_new_leaks}}
 We assume that the sequence constraints introduced by mismatch inclusion do not violate Assumption \ref{ass:orthogonality}, and that the destabilisation of duplexes does not violate Assumption \ref{ass:stability}.
\end{assumption}
In practice, mismatches will likely result in some increase in the rate of interactions between otherwise hidden toeholds; we assume that these rates remain negligible.

\begin{theorem}[Mismatches successfully destabilize unintended complexes]
 The scheme proposed in Definition \ref{def:miss1} satisfies the following:
 \begin{enumerate}
 \item  All motifs that are realisable in the mismatch-free ACDC design remain realisable in the mismatch-based scheme. 
     \item Cascades of arbitrary length $N$ with at most the first and last reactions deactivating are realisable;
     \item Feedback loops with $N$ even and $N\geq 6$ in which all reactions are activating are realisable;
     \item Feedforward loops with $N\geq 1$, $M \geq 1$, $N-M$ even, in which at most the first and last reactions are deactivating in each branch, are realisable. 
 \end{enumerate}

 \label{thm:cascade4}
\end{theorem}
The proof for Theorem \ref{thm:cascade4} is given in Appendix \ref{app:proofs}.

Note that the introduction of mismatches proposed in Definition \ref{def:miss1} invalidates Theorem \ref{thm:equivalence}, since the downstream domains of activating and deactivating catalysts are now distinct. Indeed, the described strategy only eliminates unwanted sequestration in cascades in which the intermediate steps are activating. Nonetheless, it makes complex networks in which - for example - deactivating catalysts are always active realisable. Networks of this kind are common in biology \cite{marshall_map_1994,herskowitz_map_1995}.

The first type of mismatch in Definition \ref{def:miss1} ensures that there is always a C-C mismatch between the upstream toeholds of the state strand of $A^\prime_{i+2}$ and the downstream toeholds of the state strand of $A^\prime_{i+1}$ in the cascade $A_i \rightarrow / \dashv  A_{i+1} \rightarrow A_{i+2} \rightarrow / \dashv A_{i+3}$, weakening the unwanted binding between the fuel and waste species identified in Theorem \ref{thm:cascade3}. Here $\rightarrow / \dashv $ indicates activation or deactivation. The second type of mismatch in Definition \ref{def:miss1} ensures that the upstream toeholds of the identity strand of $A_{i+2}$ are no longer fully complementary to the downstream toeholds of  $A_{i+1}$ in the cascade $A_i \rightarrow / \dashv A_{i+1} \rightarrow / \dashv A_{i+2} \rightarrow / \dashv A_{i+3}$, weakening the unwanted binding between ancillary species $A_iA_{i+1}$ and  $A_{i+2}A_{i+3}$. 

Having proposed these mismatches, it is important to determine that they would not compromise the intended reactions. The first type of mismatch in Definition \ref{def:miss1} is not present in any complex that must form during the operation of the network; only in the initially-prepared fuel and if a (de)activating catalyst binds to an (in)active substrate. It therefore presents no issues for intended reactions.    

\begin{figure}
 \centering
 \begin{subfigure}[b]{0.49 \textwidth}
 \centering
 \includegraphics[width=\textwidth]{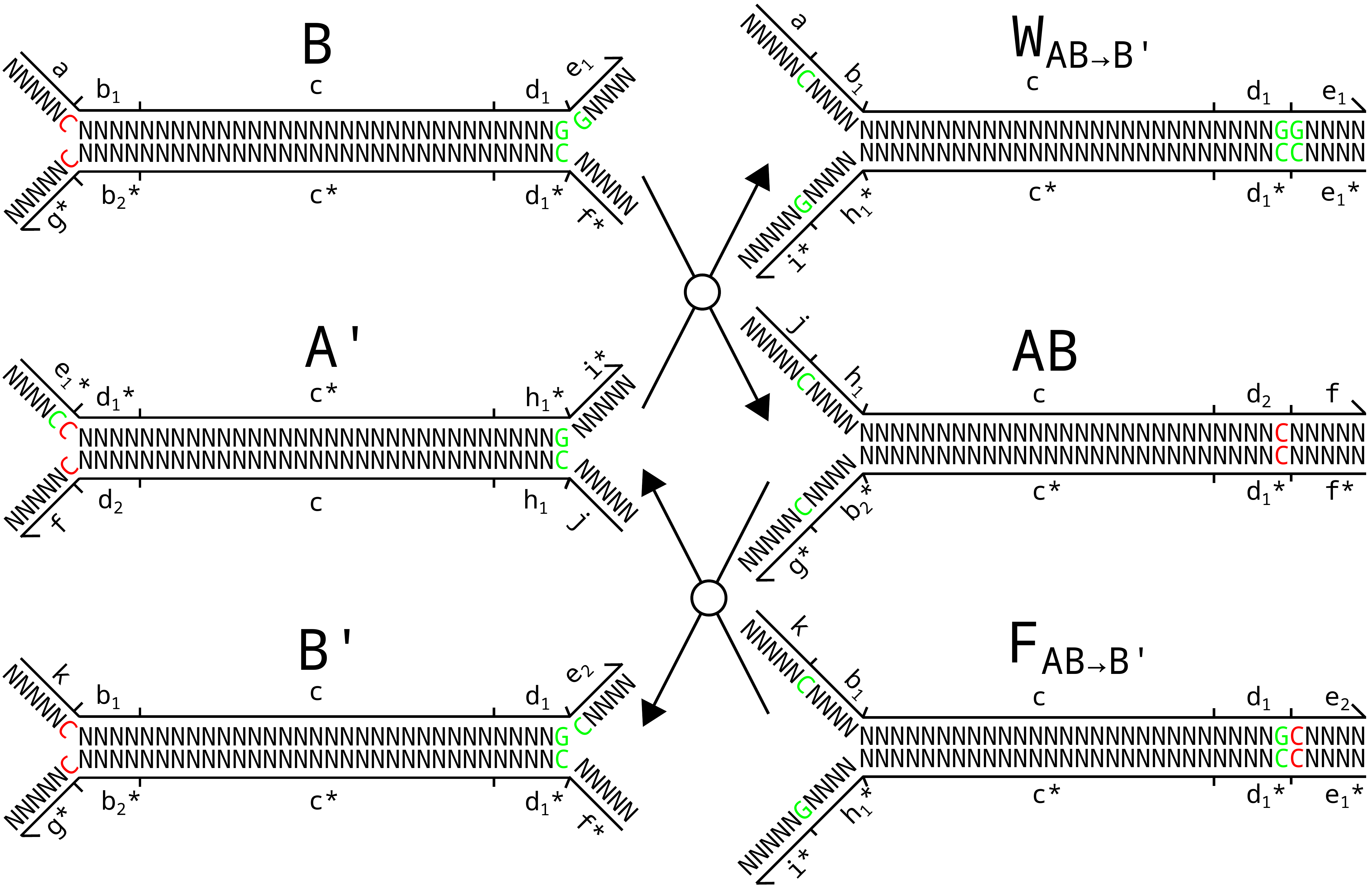}
 \caption{$A \rightarrow B$}
 \end{subfigure}
 \begin{subfigure}[b]{0.49 \textwidth}
 \centering
 \includegraphics[width=\textwidth]{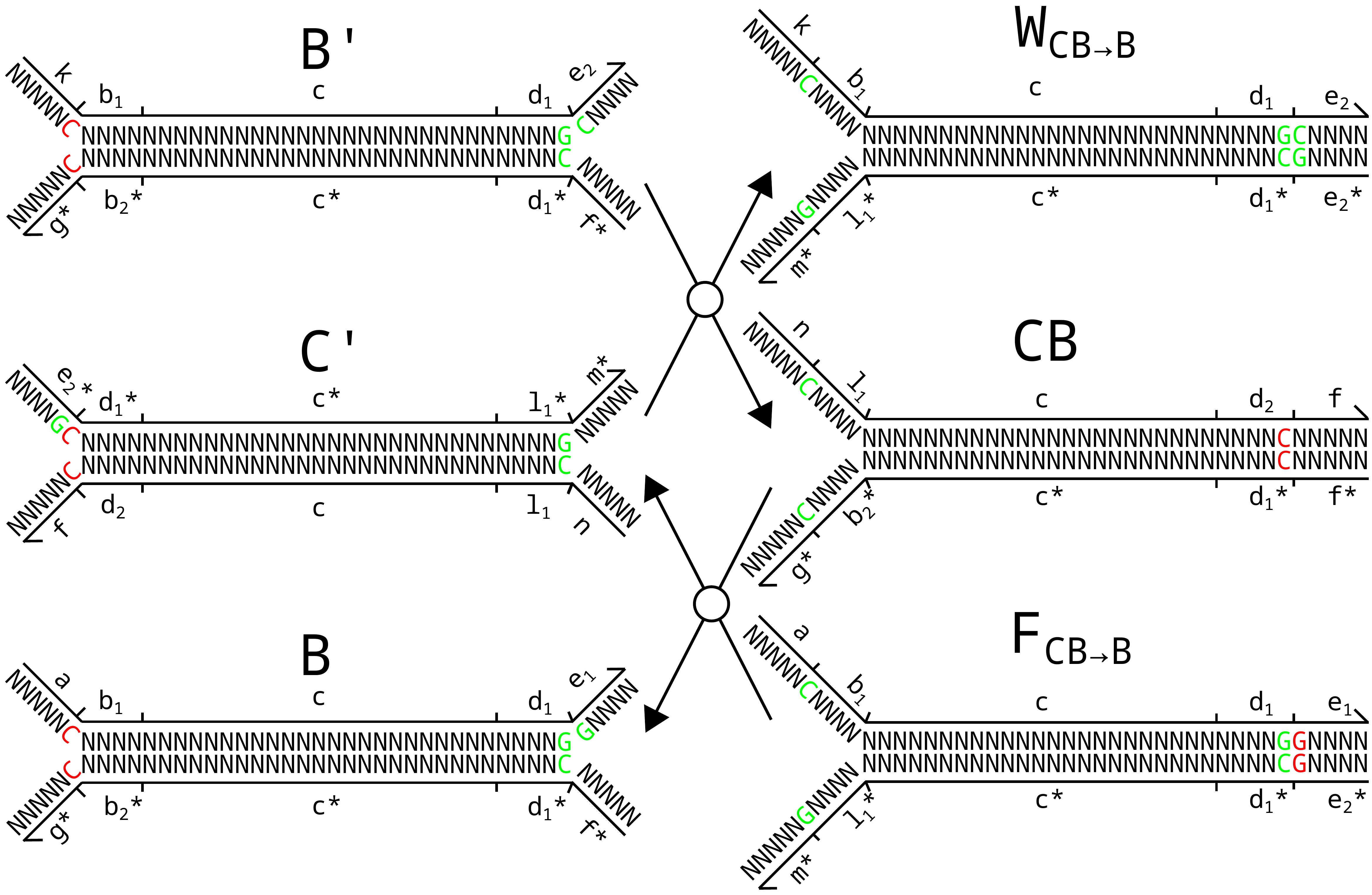}
 \caption{$C \dashv B$}
 \end{subfigure}
 \caption{Illustration of the proposed mismatch schemes for reactions $A \rightarrow B$ and $C \dashv B$, assuming toeholds of length 5 nucleotides and central domains of length 23 nucleotides. Specific mismatched bases are highlighted in red, and the same bases are highlighted in green when not part of a mismatch. 
 The domains are separated with ticks on each species, and upstream interfaces of the major species are shown on the right of each diagram.}
 \label{fig:mismatches}
\end{figure}

The second type of mismatch in Definition \ref{def:miss1} is more subtle. When a catalyst $A^\prime$ interacts with its substrate $B$, a mismatch at the very end of the catalyst duplex is converted into a mismatch within the stem of of the catalyst-substrate complex $AB$. Since mismatches are known to be more destabilizing in duplex interiors \cite{santalucia_thermodynamics_2004, ouldridge_structural_2011}, this conversion represents a local barrier to branch migration. The thermodynamic favourability of the full \textit{2r-4} reaction $A^\prime + B \rightarrow AB + W_{AB\rightarrow B^\prime}$ (or the equivalent step in a deactivation reaction) is marginal, as the mismatch at the downstream end of $B$ counters this barrier.  We assume that the local barriers introduced would not prohibit the intended reactions - indeed, conventional 3-way strand displacement is able to proceed through unmitigated C-C mismatch formation, albeit with a significant effect on kinetics \cite{machinek_programmable_2014}. In this case, any penalty is likely to be far weaker. 

The second step of the  catalytic turnover,  $AB + F_{AB\rightarrow B^\prime} \rightarrow A^\prime + B^\prime$ (or the equivalent in a deactivation) is  thermodynamically favourable (two internal mismatches are converted into exterior mismatches) and without local barriers, although one of the toeholds is effectively shortened to 4bp. The overall catalytic (de)activation cycle effectively eliminates a single C-C (G-G) mismatch initially present in the fuel. The reaction as a whole is therefore driven forwards by the free energy of base-pairing via ``hidden thermodynamic driving'' \cite{haley_design_2020}; products are more stable than reactants without consumption of initially available toeholds. In this sense, the mismatches proposed in Definition \ref{def:miss1} will improve the efficacy of the ACDC motif, as the concentration excess of fuel relative to waste required to drive the reaction in the desired direction would be reduced.

\section{A Compiler for ACDC Networks}
\label{sec:compiler}
To construct an ACDC network that implements a given graph, three things need to be done: (1) verification that  the network is realisable; (2) enumerating all domains on all species given the graph topology; and (3) compile sequences for each domain and thus for each strand present in the system. We have created an ACDC compiler with this functionality \cite{lankinen_acdc_2020}. While compilers for DSD systems that could be potentially be extended to accommodate our framework exist \cite{badelt_general-purpose_2017, srinivas_enzyme-free_2017}, we decided to make our own since our framework has unique requirements about verifying the feasibility of a given CRN and introducing mismatches within domains.

The first part is done, at least at the level of each cascade and loop present, by analysing the properties of a given graph. For every pair of nodes $i,j$, all directed simple paths are computed. We search for paths of length $N\geq 3$ that containing edge weights of -1 anywhere other than at the first or last edge; these cascades are not rendered realisable by our mismatch scheme, per Theorem \ref{thm:cascade4}. Moreover, if there exists more than 1 path between the nodes, then either a FFL (at least two paths from $i$ to $j$ or from $j$ to $i$) or a FBL (at least one path from $i$ to $j$ and from $j$ to $i$) exists in the graph. The realisability of the loop(s) can be verified from the lengths of the paths according to Theorems \ref{thm:odd_loops}, \ref{thm:self_bi} and \ref{thm:ffl}.

If a given graph is found to be realisable, then domains are assigned for each strand of each species, such that all complementarities and mismatches required by the topology are satisfied. This ask can be achieved by local analysis of the network topology. 
Finally, a NUPACK\cite{zadeh_nupack_2011} script is generated to generate optimal sequences for each strand. The required mismatches are hard-coded into the domain definitions in the script. The software is available at \url{https://doi.org/10.5281/zenodo.3838080}.

\section{Discussion}
\label{sec:discussion}
We have introduced the ACDC scheme for constructing DNA-based networks that perform direct catalysis, analysed its shortcomings, and subsequently proposed practical improvements. As of now, we have focused only on the realisibility of ACDC implementations for some graphs, not their dynamical behaviour. Three natural directions for further theoretical investigation are: (1) proving the realisability of arbitrary networks; (2) implementing additional hidden thermodynamic driving so that both \textit{2r-4} substeps of a catalytic reaction are thermodynamically downhill; and (3) automated design of ACDC networks to perform some desired transfer function between input concentrations $x_i(t),\ i=1..N$ and output concentrations $y_j(t),\ j=1..M$. With regard to the first, we conjecture that all violations of realisability in arbitrary networks are attributable to the causes identified in Section \ref{sec:networks}. 

Equally important, however, is experimentally testing the ACDC motif. Whilst 4-way branch migration has been used in several contexts \cite{venkataraman_autonomous_2007,dabby_synthetic_2013,lin_hierarchical_2018,kotani_multi-arm_2017}, the toehold exchange mechanism proposed here is relatively untested. It is also important to establish that the mismatches function as intended, limiting sequestration reactions and providing strong overall thermodynamic driving without causing excessive local barriers that frustrate the necessary reactions. A final consideration is the possibility of leak reactions involving non-complementary toeholds that we have assumed to be negligible. It remains to be established that unintended reactions will occur at a negligible rate, particularly in the context of species containing mismatches. This research is ongoing within the group.

A key property of ACDC is the two recognition interfaces within each species and the inherent symmetry in the species that follows. While this is a design feature that allows both substrate-like and catalyst-like behaviour for a single species, it also has a drawback that domains that are essential for some reaction to occur are also present in reactions where they only act as identity placeholders (downstream interface of a catalyst and an upstream interface of a substrate) that do not interact with any other domain. Consider the reaction in Figure \ref{fig:reaction}; the identity of the ``placeholder domains'' $a, b, g, h, i, j, k$ that aren't involved in the initial binding and migration reactions could be swapped to arbitrary domains that aren't complementary with $ d, e, f$ or each other in only one species and the reaction could still occur (assuming the correct fuel species is generated based on the substrate and catalyst). However, this may not be possible if $A$ and $B$ are part of some larger computation network where the placeholder domain identities are important. Another drawback of the symmetry is the limitation of loop lengths to even numbers, characterised in Theorem \ref{thm:odd_loops}. An obvious potential mitigation to this problem is to make the central domain its own complement, although this choice risks the formation of self-complementary hairpins. 

The weaknesses of the ACDC motif invite the exploration of other possible designs of catalytic activation networks that operate via direct bimolecular catalysis. It is an open question as to whether the shortcomings of ACDC can be mitigated without a substantial increase in complexity or abandoning the mechanism of direct catalytic action.

\section{Conclusion}
\label{sec:conclusion}
We have established the concept of a direct catalytic reaction and discussed why previous work on catalytic DNA computing does not fulfil this definition. We have then proposed a framework, ACDC, for implementing non-equilibrium catalytic (de)activation networks using direct catalytic activation, analogous to systems seen in living cells. ACDC is simple in the sense that all species contain only two strands - an important consideration in the context of implementing DSD circuitry in a broad range of contexts. 

 We have analysed the framework's expressiviness by exploring the implementation of seven network motifs with ACDC. The basic design is highly limited by the inherent symmetry of components, prohibiting long cascades and most feedforward and feedback loops. However, we propose that systematic placement of mismatches can obviate these difficulties in many contexts. Moreover, we argue that these initially-present mismatches can contribute a ``hidden thermodynamic driving'' \cite{haley_design_2020} to the ACDC motifs, increasing the robustness of the design to subtleties in DNA thermodynamics and reducing the concentration imbalances of fuels required to drive the reactions forward. We present a compiler for the sequence design of ACDC-based networks that implements these findings \cite{lankinen_acdc_2020}.



\bibliography{dna-nanotech}

\appendix
\section{Notation For ACDC Species and Reactions}
\paragraph*{Notation}
$[a\ b]$ denotes a strand consisting of domains $a$ and $b$. Logical not is denoted by $\lnot$ and logical and by $\land$. $\{n..m\}$,  with $n < m$, denotes the integer interval between $n$ and $m$.
\paragraph*{Definitions}
\begin{definition}{(ACDC reactant structure).}
 Each reactant in an ACDC network consists of two strands, each of which have one long domain and four toehold domains. The two strands are called \textit{state strand} and \textit{identity strand} based on the fact that one strand decodes the state of the species and other the identity. A reactant $X$ has the following domains (note the use of $H$ for ``inner'' to avoid confusion with ``identity'':
 \begin{itemize}
  \item $SH5(X)$: the inner toehold domain on the 5' side (downstream end) of the state strand.
  \item $SO5(X)$: the outer toehold domain on the 5' side (downstream end) of the state strand.
  \item $SH3(X)$: the inner toehold domain on the 3' side (upstream end) of the state strand.
  \item $SO3(X)$: the outer toehold domain on the 3' side (upstream end) of the state strand.
  \item $IH5(X)$: the inner toehold domain on the 5' side (upstream end) of the identity strand.
  \item $IO5(X)$: the outer toehold domain on the 5' side (upstream end) of the identity strand.
  \item $IH3(X)$: the inner toehold domain on the 3' side (downstream end) of the identity strand.
  \item $IO3(X)$: the outer toehold domain on the 3' side (downstream end) of the identity strand.
  \item $SL(X)$: the long domain on the state strand.
  \item $IL(X)$: the long domain on the identity strand.
 \end{itemize}
\end{definition}

\begin{definition}{(Subset and logical operations for ACDC species).}
 The following operations will be useful in the analysis of ACDC networks:
 \begin{itemize}
  \item Complementarity $\diamond$ : $x \diamond y$ is true for sequences $x, y$ iff $x = y^*$ (and $x^* = y$).
  \item Complementarity with mismatch $\square$ : $x \square y$ is true for sequences $x, y$ iff $x = y^*$ (and $x^* = y$) except for a single centrally-placed C-C or G-G mismatch. $x \square y$ is distinct from $\lnot x \diamond y$, for which it is assumed that interactions between $x$ and $y$ are negligible. 
  \item $5^\prime$ (downstream end) state toehold sequence $S5(A) := [SO5(A)\ SH5(A) ]$.
  \item $3^\prime$ (upstream end) state toehold sequence $S3(A) := [SH3(A)\ SO3(A) ]$.
  \item $5^\prime$ (upstream end) identity toehold sequence $I5(A) := [IO5(A)\ IH5(A) ]$.
  \item $3^\prime$ (downstream end) identity toehold sequence $I3(A) := [IH3(A)\ IO3(A) ]$.
 \end{itemize}
\end{definition}

\begin{definition}{(Major species).}{\label{def:species}}
 A species $X$ is either a major species  only if
 \begin{align}
  \lnot & \big(SO5(X) \diamond IO3(X)\big) \land \nonumber \\
  & \big(SH5(X) \diamond IH3(X)\big) \land \nonumber \\
  & \big(SL(X) \diamond IL(X)\big) \land \nonumber \\
  & \big(SH3(X) \diamond IH5(X)\big) \land \nonumber \\
  \lnot & \big(SO3(X) \diamond IO5(X)\big). \nonumber
 \end{align}
\end{definition}

\begin{definition}{(Domain complementarities in an ACDC reaction without mismatches).}{\label{def:domaincomp}}
 An ACDC reaction
 \begin{align}
  A \rightarrow B \nonumber
 \end{align}
 or 
 \begin{align}
  A \dashv B \nonumber
 \end{align}
 implies
 \begin{align}
  & S5(A') \diamond S3(B) =S3(B^\prime) &\land \nonumber \\
  & IL(A') =IL(A) \diamond IL(B)=IL(B^\prime) &\land \nonumber \\
  & I3(A') =I3(A) \diamond I5(B) =I5(B^\prime). \nonumber
 \end{align}
 Domains not constrained by these requirements are non-complementary.
\end{definition}

\begin{definition}{(Domain complementarities in ACDC reactions with mismatches).}{\label{def:mismatch_domaincomp}}
 An ACDC reaction \ref{def:miss1}
 \begin{align*}
  A \rightarrow B
 \end{align*}
 with mismatches placed as per Definition \ref{def:miss1} implies
 \begin{align*}
  & S5(A') \diamond S3(B) &\land \\
  & S5(A') \square S3(B') &\land \\
  & IL(A') =IL(A) \diamond IL(B) =IL(B^\prime) &\land \\
  & I3(A') =I3(A) \square I5(B) = I5(B^\prime).
 \end{align*}
  Domains not constrained by these requirements are non-complementary.

 An ACDC reaction
 \begin{align*}
  A \dashv B
 \end{align*}
 with mismatches placed as per Definition \ref{def:miss1} implies
 \begin{align*}
  & 5(A') \square S3(B) &\land \\
  & S5(A') \diamond S3(B') &\land \\
  & IL(A') =IL(A) \diamond IL(B) =IL(B^\prime) &\land \\
  & I3(A') =I3(A) \square I5(B) = I5(B^\prime).
 \end{align*}
  Domains not constrained by these requirements are non-complementary.
\end{definition}

\section{Proofs of Theorems and Lemmas \ref{thm:split} - \ref{thm:cascade4}}
\label{app:proofs}
\setcounter{theorem}{7}
\begin{theorem}[Split motifs are realisable]
 Consider the set of $N$ reactions
 \begin{align*}
   A \rightarrow B_1 \hspace{5mm}
   A \rightarrow B_2 \hspace{5mm}
 \hdots \hspace{5mm}
 &A \rightarrow B_N,
 \end{align*}
 in which all $B_i$ are distinct nodes from $A$.
 Such a network is realisable for any $N \geq 1$.
\end{theorem}
\begin{proof}
 To realise the above system, we must have:
 \begin{itemize}
     \item Definition \ref{def:species} must apply for $A^\prime$ and $B_i$ for all $i$,
     \item Definition \ref{def:species} must apply for all pairs $A^\prime$, $B_i$,
     \item  $S5(B_i)$ and and $I3(B_i)$ must be unique for all $i$.
     \nonumber \\
 \end{itemize}
 By simple inspection it can be verified that there is no contradiction in these requirements. Moreover, all toeholds other than those required to be complementary can be chosen to be non-complementary to each other. If these assignments are made, it can be directly verified that Definition \ref{def:realisable} is not violated by the major species and associated ancillary species. Thus the motif is realisable.
\end{proof}
\begin{theorem}[Integrate motifs are realisable]{\label{thm:integrate}}
 Consider the set of $N$ reactions
 \begin{align*}
   A_1 \rightarrow B \hspace{5mm}
   A_2 \rightarrow B \hspace{5mm}
 \hdots \hspace{5mm}
 &A_N \rightarrow B,
 \end{align*}
 in which all $A_i$ are distinct nodes from $B$. Such a system is realisable
 for any $N \geq 1$.
\end{theorem}
\begin{proof}
To realise the above system, we must have:
 \begin{itemize}
     \item Definition \ref{def:species} must apply for $A^\prime_i$ and $B$ for all $i$,
     \item Definition \ref{def:species} must apply for all pairs $A^\prime_i$, $B$.
     \nonumber \\
 \end{itemize}
By simple inspection it can be verified that there is no contradiction in these requirements. Moreover, all toeholds other than those required to be complementary can be chosen to be non-complementary to each other. If these assignments are made, it can be directly verified that Definition \ref{def:realisable} is not violated by the major species and associated ancillary species. Thus the motif is realisable.
\end{proof}
\begin{lemma}[The ancillary species of a catalyst's upstream reactions and substrate's  downstream reactions cause leak reactions]
Consider a reaction $B \rightarrow C$, and further assume that $A \rightarrow B$ and $C \rightarrow D$ for at least one species $A$ and at least one species $D$. Then  $AB$ and $CD$, and $F_{AB\rightarrow B^\prime}$ and $F_{CD\rightarrow D^\prime}$/$W_{CD\rightarrow D^\prime}$ possess two available toehold pairs that could form a contiguous complementary duplex. No other violations of realisability occur.
\end{lemma}
\begin{proof}
It can be verified by inspection that there is no inconsistency in the domain requirements for $A, B, C, D$ to be defined as major species (Definition \ref{def:species}) and for $A \rightarrow B \rightarrow C \rightarrow D$ (Definition \ref{def:domaincomp}). All domains can be chosen to be non-complementary unless specified by these requirements. When these domain assignments are made, it can be verified by inspection that criteria 1, 2 and 3 of Definition \ref{def:realisable} are not violated. 

To establish whether criterion 4 of Definition \ref{def:realisable} is violated, one need only consider the unbound domains on the ancillary species in the system $A \rightarrow B \rightarrow C \rightarrow D$:
\begin{itemize}
 \item $I5(A),\ I3(B)$ in $AB$
 \item $S3(A'),\ S5(B')$ in $F_{AB \rightarrow B'}$
 \item $S3(A'),\ S5(B)$ in $W_{AB \rightarrow B'}$
 \item $I5(B),\ I3(C)$ in $BC$
 \item $S3(B'),\ S5(C')$ in $F_{BC \rightarrow C'}$
 \item $S3(B'),\ S5(C)$ in $W_{BC \rightarrow C'}$
 \item $I5(C),\ I3(D)$ in $CD$
 \item $S3(C'),\ S5(D')$ in $F_{CD \rightarrow D'}$
 \item $S3(C'),\ S5(D)$ in $W_{CD \rightarrow D'}$.
\end{itemize}
Observe that the reaction $B \rightarrow C$ implies $I3(B) \diamond I5(C),\ S5(B') \diamond S3(C)$, meaning $AB$ and $BC$ can bind by the two  contiguous toehold domains $I3(B),\ I5(C)$, and $F_{AB \rightarrow B'}$ can bind with $F_{CD \rightarrow D'}$ and $W_{CD \rightarrow D'}$ by the two contiguous toehold domains in $S5(B'),\ S3(C)$. It can be verified by inspection that no other violations of criterion 4 occur. 
\end{proof}
\setcounter{theorem}{12}
\begin{theorem}[Long cascades are non-realisable due to a particular type of leak reaction only]
Consider the set of $N$ reactions
$A_1 \rightarrow A_{2},\, A_{2} \rightarrow A_{3} \,...\, A_{N-1}\rightarrow A_N$, in which all $A_i$ are distinct. This network would be realisable if reactions between ancillary species $A_iA_{i+1}$ and $A_{i+2}A_{i+3}$, and  $F_{A_{i}A_{i+1}\rightarrow A_{i+1}^\prime}$ and $F_{A_{i+2}A_{i+3}\rightarrow A_{i+3}^\prime}$/$W_{A_{i+2}A_{i+3}\rightarrow A_{i+3}^\prime}$, were absent. 
\end{theorem}
\begin{proof}
It can be directly verified at that an arbitrarily-long cascade can be constructed at the domain level in which each $A_i$ satisfies Definition \ref{def:species} and each pair $A_i,A_{i+1}$ satisfies Definition \ref{def:domaincomp}, satisfying criterion 1 of Definition \ref{def:realisable}.

If  all sequences not constrained to be complementary by these definitions are chosen to be non-complementary, potential violations of the criteria 2-4 of Definition \ref{def:realisable} arise due to an unavoidable unwanted complementarity between toehold domains in species that are not intended to interact. By explicitly constructing a cascade at the domain level, it can be verified that some toehold domains (or their complements) present in the reaction $A_i \rightarrow A_{i+1}$ \textbf{must} also be present in $A_{i+1} \rightarrow A_{i+2}$ and $A_{i+2} \rightarrow A_{i+3}$, but not in  $A_{i+n} \rightarrow A_{i+n+1}$ for $n\geq 3$. Intuitively, strands that participate at cascade level $j$ also participate at level $j-1$ or $j+1$, but no further away.  It is therefore sufficient to consider a cascade with $N=4$ to identify all violations of realisability in a cascade. The required result then follows directly from Lemma \ref{lemma:leak1}.
\end{proof}

\setcounter{theorem}{14}
\begin{theorem}[Self interactions and bidirectional edges are not realisable]
Consider a system of reactions $A_1 \rightarrow A_2 \rightarrow A_3 \hdots A_{N-1} \rightarrow A_1$. This network is not realisable if $N\leq2$.
\end{theorem}
\begin{proof}
The result for $N=1$ is a direct consequence of Theorem \ref{thm:odd_loops}. For $N=2$, consider the set of reactions
\begin{align*}
  &A \rightarrow B \\
  &B \rightarrow A. 
 \end{align*}
 By Definition \ref{def:domaincomp}, $A \rightarrow B$ implies $I3(A) \diamond I5(B)$ and $IL(A) \diamond IL(B)$. In addition, $B \rightarrow A$ implies $I5(A) \diamond I3(B)$. The identity strands of $A$ and $B$ are then fully complementary, violating criterion 3 of Definition \ref{def:realisable}.
\end{proof}
\setcounter{theorem}{16}
\begin{theorem}[Long feedback loops with an even number of units are non-realisable due to a particular type of leak reaction only]
{\label{thm:fbl}}
  Consider the feedback loop
 \begin{align*}
  &A_1 \rightarrow A_2 \hspace{5mm} A_2 \rightarrow A_3 \hspace{5mm} \hdots \hspace{5mm} A_{N-1} \rightarrow B_N \hspace{5mm} A_N \rightarrow A_1 
  \end{align*}
For $N$ even, $N\geq 4$, this network would be realisable if reactions between ancillary species $A_iA_{i+1}$ and $A_{i+2}A_{i+3}$, and  $F_{A_{i}A_{i+1}\rightarrow A_{i+1}^\prime}$ and $F_{A_{i+2}A_{i+3}\rightarrow A_{i+3}^\prime}$/$W_{A_{i+2}A_{i+3}\rightarrow A_{i+3}^\prime}$,  were absent.  Here, the index $j$ in $A_j$ should be interpreted modularly: $A_j = A_{j-N}$ for $j>N$.       
\end{theorem}
\begin{proof}
 It can be directly verified at that for $N$ even, $N\geq 4$, a loop can be constructed at the domain level in which each $A_i$ satisfies Definition \ref{def:species} and each pair $A_i,A_{i+1}$ (defined modularly) satisfies Definition \ref{def:domaincomp}, satisfying criterion 1 of Definition \ref{def:realisable}.

To identify the violations of realisability that arise from unwanted interactions, let us first consider a cascade without the $A_N \rightarrow A_1 $ reaction. The only violations of realisability are those identified in \ref{thm:cascade3}: between $A_iA_{i+1}$ and $A_{i+2}A_{i+3}$, and  $F_{A_{i}A_{i+1}\rightarrow A_{i+1}^\prime}$ and $F_{A_{i+2}A_{i+3}\rightarrow A_{i+3}^\prime}$/$W_{A_{i+2}A_{i+3}\rightarrow A_{i+3}^\prime}$, without interpreting the index modularly. Now we consider the additional effect of requiring $A_N \rightarrow A_1$. The only  domains that must be changed are $S5(A_N^\prime)$ and $I3(A_N^\prime)$. These domains and their complements are only present in reactions $A_{N-1} \rightarrow A_N$, $A_N \rightarrow A_1$, $A_1 \rightarrow A_2$, and so it is sufficient to consider only this cascade to identify additional violations of realisability. By Lemma \ref{lemma:leak1}, the resultant violations of realisability are exactly those stated in the theorem. 
\end{proof}

\begin{theorem}[The relative lengths of paths are constrained in feedforward loops]{\label{thm:ffl}}
 Consider the generalised feedforward loop
 \begin{align*}
  &A \rightarrow B_1 \hspace{5mm} B_1 \rightarrow B_2 \hspace{5mm} \hdots \hspace{5mm} B_{N-1} \rightarrow B_N \hspace{5mm} B_N \rightarrow D \\
  &A \rightarrow C_1 \hspace{5mm} C_1 \rightarrow C_2 \hspace{5mm} \hdots \hspace{5mm} C_{M-1} \rightarrow C_M \hspace{5mm} C_M \rightarrow D 
 \end{align*}
 For such a network to be realisable, it is necessary that $N \geq 1$, $M \geq 1$, and $N-M$ is even. 
\end{theorem}
\begin{proof}
 The claim about $N-M$ having to be even follows from Theorem \ref{thm:odd_loops}.
 
 Assume for contradiction that a FFL with $N=0$ and $M \geq 2$ and even is realisable. Since $A$ activates $C_1$, and both $A$ and $C_M$ activate $D$, it must be that $C_M$ can also perform a branch migration with $C_1$, which is an unwanted reaction violating criterion 2 of Defintion \ref{def:realisable}.\par
\end{proof}

\setcounter{theorem}{22}
\begin{theorem}[Mismatches successfully destabilize unintended complexes]
 The scheme proposed in Definition \ref{def:miss1} satisfies the following:
 \begin{enumerate}
 \item  All motifs that are realisable in the mismatch-free ACDC design remain realisable in the mismatch-based scheme. 
     \item Cascades of arbitrary length $N$ with at most the first and last reactions deactivating are realisable;
     \item Feedback loops with $N$ even and $N\geq 6$ in which all reactions are activating are realisable;
     \item Feedforward loops with $N\geq 1$, $M \geq 1$, $N-M$ even, in which at most the first and last reactions are deactivating in each branch, are realisable. 
 \end{enumerate}
\end{theorem}
\begin{proof}
 Consider the first claim. For any network in which it is possible to select domains that satisfy Definition \ref{def:species} and Definition \ref{def:domaincomp}, it is trivial to convert those domains to satisfy \ref{def:species} and \ref{def:domaincomp} by introducing the specific bases at the required locations in the major species, and adjusting the fuel and waste to compensate. By Assumption \ref{ass:mismatches}, these changes do not introduce new violations of realisability.
 
Now consider the second claim. By the first claim and \ref{thm:cascade3}, it is sufficient to consider whether the sequestration reactions characterised by Lemma \ref{lemma:leak1} occur between ancillary species in any cascade of $N=4$ components in the mismatch-based scheme of Definition \ref{def:miss1}.
 
 First, consider the unbound domains in the ancillary species in the system $A \rightarrow / \dashv B \rightarrow C \rightarrow / \dashv D$, with mismatches placed as per Definition \ref{def:miss1}:
\begin{itemize}
 \item $I5(A),\ I3(B)$ in $AB$
 \item $S3(A'),\ S5(B')$ in $F_{AB \rightarrow B'}\ W_{AB \rightarrow B}$
 \item $S3(A'),\ S5(B)$ in $W_{AB \rightarrow B'}\ F_{AB \rightarrow B}$
 \item $I5(B),\ I3(C)$ in $BC$
 \item $S3(B'),\ S5(C')$ in $F_{BC \rightarrow C'}$
 \item $S3(B'),\ S5(C)$ in $W_{BC \rightarrow C'}$
 \item $I5(C),\ I3(D)$ in $CD$
 \item $S3(C'),\ S5(D')$ in $F_{CD \rightarrow D'}\ W_{CD \rightarrow D}$
 \item $S3(C'),\ S5(D)$ in $W_{CD \rightarrow D'}\ F_{CD \rightarrow D}$.
\end{itemize}

By Definition \ref{def:mismatch_domaincomp}, observe that  the reaction $B \rightarrow C$ implies $I3(B) \square I5(C),\ S5(B') \square S3(C')$. Moreover, $\lnot S5(B) \diamond S3(C')$. By Assumption \ref{ass:mismatches}, none of the violations of realisability that would otherwise occur due to binding of $AB$ and $CD$; $F_{AB \rightarrow B'}\ W_{AB \rightarrow B}$ and $F_{CD \rightarrow D'}\ W_{CD \rightarrow D}$; and $F_{AB \rightarrow B'}\ W_{AB \rightarrow B}$ and $W_{CD \rightarrow D'}\ F_{CD \rightarrow D}$ characterised by Lemma \ref{lemma:leak1}, occur.  

We note that if $B \dashv C$ in the above network, Definition \ref{def:mismatch_domaincomp} implies $S5(B') \diamond S3(C')$, meaning that sequestration reactions still occur between ancillary fuel and waste species. Cascades with deactivation reactions as intermediate steps are therefore not realisable in this scheme.



Now consider the third claim. By Theorem \ref{thm:fbl} and the first claim of this Theorem, it is sufficient to consider only the sequestration reactions listed in Theorem \ref{thm:fbl}. Further, since the only difference between a feedback loop with exclusively activating interactions and an activating cascade with $N$ species is that  $A_{N} \rightarrow A_1$ in a loop, the second claim of this Theorem implies that it is only necessary to consider changes in realisability due to the introduction of the $A_{N} \rightarrow A_1$ reaction.

For $N\geq 6$, it can be verified that imposing $I3(A_N) \square I5(A_1)$, $S5(A_N^\prime) \square S(3) A_1^\prime$, as required by $A_{N} \rightarrow A_1$, does not create new realisability violations for a cascade of length $N$ with exlcusively activating reactions. The ancillary species of the reactions $A_{N-2}\rightarrow A_{N-1}, A_{N-1}\rightarrow A_N, A_N\rightarrow A_{1}, A_1 \rightarrow A_2, A_2 \rightarrow A_3$ can only form complexes held together by two contiguous toehold domains with a central mismatch, and thus do not violate realisability by Assumption \ref{ass:mismatches}. All other ancillary species are unaffected. 

We note that the above argument does not apply to FBLs of length $N=4$, which remain unrealisable. In that case, adding the reaction $A_{N} \rightarrow A_1$ creates complexes of ancillary species that are held together by two separate sets of contiguous toehold domains, each with a central mismatch, either side of a 4-way junction. In effect, the short periodicity of an $N=4$ loop means that the unwanted interaction identified in Lemma \ref{lemma:leak1} happens twice for each pair of ancillary species. We do not assume in Assumption \ref{ass:mismatches} that such a structure will dissociate. We also note that feedback loops with any deactivating reactions remain unrealisable, since each reaction $A_i \rightarrow A_{i+1}$ is effectively an intermediate reaction between  $A_{i-1} \rightarrow A_i$ and $A_{i+1} \rightarrow A_{i+2}$.



Finally we turn to the fourth claim. By the first claim of this Theorem, and Theorem \ref{thm:ffl}, it is sufficient to consider only the potential unwanted sequestration reactions between ancillary species identified in Theorem \ref{thm:ffl} for each feed-forward branch. The proof is then identical to that of the second claim of this Theorem.
\end{proof}

\end{document}